\documentclass[a4paper, 11pt]{article}
\usepackage[utf8]{inputenc}
\usepackage[ngermanb, english]{babel}
\usepackage{amsmath}
\usepackage{amssymb}
\usepackage{amsthm}
\usepackage{tikz-cd}
\usepackage{bm}
\usepackage{commath}
\usepackage{mathrsfs}
\usepackage{mathtools}
\usepackage{slashed}
\usepackage{tensor}
\usepackage{parskip}
\usepackage{epstopdf}
\usepackage{graphicx}
\usepackage[nottoc]{tocbibind}
\usepackage{babelbib}
\usepackage{hyperref}
\usepackage[top=2.5cm, left=2.5cm, right=2.5cm, bottom=2.5cm]{geometry}

\theoremstyle{plain}
\newtheorem{thm}{Theorem}[section]
\newtheorem{lem}[thm]{Lemma}
\newtheorem{prop}[thm]{Proposition}

\newtheorem{col}[thm]{Corollary}

\theoremstyle{definition}
\newtheorem{defn}[thm]{Definition}

\newtheorem{ter}[thm]{Terminology}

\newtheorem{exmp}[thm]{Example}

\theoremstyle{remark}
\newtheorem{rem}[thm]{Remark}

\providecommand{\sectionref}[1]{Section~\ref{#1}}

\providecommand{\eqnref}[1]{Equation~\eqref{#1}}
\providecommand{\eqnsref}[1]{Equations~\eqref{#1}}

\providecommand{\exmpsaref}[2]{Examples~\ref{#1} and \ref{#2}}
\providecommand{\defnsaref}[2]{Definitions~\ref{#1} and \ref{#2}}

\DeclareSymbolFont{extraitalic}      {U}{zavm}{m}{it}
\DeclareMathSymbol{\Qoppa}{\mathord}{extraitalic}{161}
\DeclareMathSymbol{\qoppa}{\mathord}{extraitalic}{162}
\DeclareMathSymbol{\Stigma}{\mathord}{extraitalic}{167}
\DeclareMathSymbol{\Sampi}{\mathord}{extraitalic}{165}
\DeclareMathSymbol{\sampi}{\mathord}{extraitalic}{166}
\DeclareMathSymbol{\stigma}{\mathord}{extraitalic}{168}

\newlength{\graphlength}
\newlength{\cgreenlength}
\newcommand{\deDonder}{{d \negmedspace D \mspace{-2mu}}}

\newcommand{\one}{\mathbb{I}}
\newcommand{\coone}{\hat{\mathbb{I}}}
\newcommand{\Q}{{\mathbf{Q}}}
\newcommand{\HQ}{\mathcal{H}_{\Q}}
\newcommand{\HQd}{\boldsymbol{\mathcal{H}}_{\Q}}
\newcommand{\EQ}{{\mathcal{E}_{\Q}}}
\newcommand{\RQ}{\mathcal{R}_{\Q}}
\newcommand{\RQO}{{\mathcal{R}_{\Q}^{[0]}}}
\newcommand{\RQI}{{\mathcal{R}_{\Q}^{[1]}}}
\newcommand{\AQ}{\mathcal{A}_{\Q}}

\newcommand{\DQ}[1]{\mathcal{D} \left ( #1 \right )}
\newcommand{\EG}[1]{\mathcal{E} \left ( #1 \right )}
\newcommand{\DQprime}[1]{\mathcal{D}^\prime \left ( #1 \right )}

\newcommand{\CQ}{\boldsymbol{\mathcal{C}}_{\Q}}
\newcommand{\PQ}{\boldsymbol{\mathcal{P}}_{\Q}}
\newcommand{\FQ}{\mathcal{F}_{\Q}}
\newcommand{\GQ}{\mathcal{G}_{\Q}}
\newcommand{\GQc}{\boldsymbol{\mathcal{G}}_{\Q}}
\newcommand{\GQt}{\boldsymbol{\mathcal{G}}^\text{tree}_{\Q}}
\newcommand{\GQl}{\boldsymbol{\mathcal{G}}^\text{loop}_{\Q}}
\newcommand{\GQo}{\boldsymbol{\mathcal{G}}^{\sigma, \tau}_{\Q}}
\newcommand{\LQ}{\mathcal{L}_{\Q}}
\newcommand{\TQ}{\mathcal{T}_{\Q}}

\newcommand{\VQ}{\mathcal{V}_{\Q}}
\newcommand{\VQd}{\boldsymbol{\mathcal{V}}_{\Q}}

\newcommand{\Det}[1]{\operatorname{Det} \left ( #1 \right )}

\newcommand{\sym}[1]{\operatorname{Sym} \left ( #1 \right )}
\newcommand{\aut}[1]{\operatorname{Aut} \left ( #1 \right )}
\newcommand{\res}[1]{\operatorname{Res} \left ( #1 \right )}

\newcommand{\cplgrd}[1]{\operatorname{CplGrd} \left ( #1 \right )}
\newcommand{\vtxgrd}[1]{\operatorname{VtxGrd} \left ( #1 \right )}

\newcommand{\dt}[1]{\operatorname{Det} \left ( #1 \right )}

\newcommand{\sdd}[1]{\omega \left ( #1 \right )}

\newcommand{\mult}{m}
\newcommand{\ring}{\Bbbk}

\newcommand{\regFR}{{\Phi_\mathscr{E}^\varepsilon}}

\newcommand{\countertermsymbol}{S_\mathscr{R}^\regFR}
\newcommand{\counterterm}[1]{\countertermsymbol \left ( #1 \right )}
\newcommand{\renscheme}[1]{\mathscr{R} \left ( #1 \right )}

\newcommand{\id}{\operatorname{Id}}
\newcommand{\precombgreen}{\mathfrak{x}}
\newcommand{\combgreen}{\mathfrak{X}}
\newcommand{\rescombgreen}{\mathfrak{X}}
\newcommand{\precombgreenc}{\mathfrak{y}}
\newcommand{\combgreenc}{\mathfrak{Y}}
\newcommand{\rescombgreenc}{\mathfrak{Y}}
\newcommand{\combcharge}{\mathfrak{Q}}

\newcommand{\ZvQ}{\mathbb{Z}^{\mathfrak{v}_\Q}}
\newcommand{\ZqQ}{\mathbb{Z}^{\mathfrak{q}_\Q}}
\newcommand{\surject}{\to \!\!\!\!\! \to}

\newcommand{\enter}{\vspace{\baselineskip}}
\newcommand{\mathbbit}[1]{{\mspace{1mu} \italicbox{$\mathbb{#1}$} \mspace{2mu}}}

\newcommand{\bbL}{\mathbbit{L}}
\newcommand{\bbI}{\mathbbit{I} \mspace{1mu}}
\newcommand{\bbT}{\mathbbit{T} \mspace{1mu}}

\newcommand{\GGamma}{\boldsymbol{\Gamma}}
\newcommand{\ggamma}{\boldsymbol{\gamma}}
\newcommand{\llambda}{\boldsymbol{\lambda}}

\newcommand{\gtproj}{\vcenter{\hbox{\scalebox{1.5}{\(\blacktriangleright\)}}}}
\newcommand{\gfproj}{\vcenter{\hbox{\scalebox{1.5}{\(\vartriangleright\)}}}}
\newcommand{\oscan}{\vcenter{\hbox{\scalebox{1.5}{\(\vert\)}}}}

\newcommand{\graph}[2]{\vcenter{\hbox{\includegraphics[width=#1]{#2}}}}

\newsavebox{\foobox}
\newcommand{\italicbox}[2][.25]
{%
	\mbox
	{%
		\sbox{\foobox}{#2}%
		\hskip\wd\foobox
		\pdfsave
		\pdfsetmatrix{1 0 #1 1}%
		\llap{\usebox{\foobox}}%
		\pdfrestore
	}%
}

\makeatletter
\newcommand{\subalign}[1]{%
  \vcenter{%
    \Let@ \restore@math@cr \default@tag
    \baselineskip\fontdimen10 \scriptfont\tw@
    \advance\baselineskip\fontdimen12 \scriptfont\tw@
    \lineskip\thr@@\fontdimen8 \scriptfont\thr@@
    \lineskiplimit\lineskip
    \ialign{\hfil$\m@th\scriptstyle##$&$\m@th\scriptstyle{}##$\crcr
      #1\crcr
    }%
  }
}
\makeatother

\makeatletter
\newcommand{\oset}[3][0ex]{%
  \mathrel{\mathop{#3}\limits^{
    \vbox to#1{\kern-2\ex@
    \hbox{$\scriptstyle#2$}\vss}}}}
\makeatother

\makeatletter
\newcommand{\uset}[3][0ex]{%
  \mathrel{\mathop{#3}\limits_{
    \vbox to#1{\kern-7\ex@
    \hbox{$\scriptstyle#2$}\vss}}}}
\makeatother

\title{\textsc{Cancellation Identities and Renormalization}}
\author{David Prinz\footnote{Max Planck Institute for Mathematics, Bonn; e-mail:\ \href{mailto:prinz@mpim-bonn.mpg.de}{prinz@mpim-bonn.mpg.de}, website: \href{https://davidprinz.org}{davidprinz.org}}}
\date{December 1, 2025}

\begin{document}

\maketitle

\begin{abstract}
	We construct a manifest gauge invariant renormalization framework by first introducing a perturbative BRST Feynman graph complex and then combining it with Connes--Kreimer renormalization theory: To this end, we first formalize the cancellation identities of 't Hooft (1971), which were used to prove the absence of gauge anomalies in Quantum Yang--Mills theories. Specifically, we start with some reasonable axioms of (generalized) gauge theories and then present the most general version of cancellation identities ensuring transversality. Then, we construct a perturbative BRST Feynman graph complex, whose cohomology groups consist of transversal invariant linear combinations of Feynman graphs. We prove that the cohomology groups are zero in odd degree and generated by connected combinatorial Green's functions in even degree, with a corresponding number of external ghost edges. Ultimately, we then formulate the renormalization Hopf algebra on these cohomology groups, which directly links to Hopf subalgebras for multiplicative renormalization. Finally, we exemplify the developed theory with Quantum Yang--Mills theory and (effective) Quantum General Relativity.
\end{abstract}

\section{Introduction} \label{sec:introduction}

A central problem in the quantization of (generalized) gauge theories is the avoidance of gauge anomalies. More specifically, if possible, the renormalization operation needs to be constructed in such a way that the multiplicatively renormalized Lagrange density is still invariant under the (generalized) gauge symmetry. In this article, we propose a manifestly gauge invariant renormalization framework based on the \emph{Hopf algebraic renormalization} by Connes and Kreimer \cite{Kreimer_Hopf_Algebra,Connes_Kreimer_NG,Broadhurst_Kreimer} and the \emph{cancellation identities} of 't Hooft and Veltman \cite{tHooft,tHooft_Veltman,Cvitanovic}. For us, a (generalized) gauge theory is a theory which admits a \emph{BRST symmetry} on its space of fields: Notably, this includes General Relativity with its diffeomorphism invariance, cf.\ \cite{Faizal,Prinz_5,Prinz_6} and the references therein. In particular, this setting is sufficient to study the transversality of the corresponding Feynman rules by means of a longitudinal projection operator and thus allows for a diagrammatic description thereof, cf.\ \cite{Prinz_4,Kissler,Prinz_7} and the cited sources. Explicitly, the BRST symmetry is a supersymmetry transformation generated by the BRST operator \(S\), which sends a field \(\varphi\) to an infinitesimal gauge transformation thereof in the direction of a ghost field \(\theta\), i.e.\ \(S \varphi \coloneq \delta_\theta \varphi\), and then extended to the ghost field such that it satisfies the nilpotency condition \(S^2 = 0\). Notably, a theory \(\Q\) is BRST invariant if its Lagrange density \(\LQ\) is essentially BRST closed, i.e.\ \(S \LQ \simeq_\text{TD} 0\), where \(\simeq_\text{TD}\) means equality modulo total derivatives. Now, for multiplicative renormalization we start with the Lagrange density and decompose it into its different monomials \(\LQ \equiv \sum_{r \in \RQ} \mathcal{M}^r\), where \(r\) denotes the amplitude and \(\mathcal{M}^r\) its corresponding monomial. Then, we associate to each such amplitude \(r\) its counterterm \(Z\)-factor \(Z^r\), such that the multiplicatively renormalized Lagrange density is given by \(\LQ^\texttt{R} \equiv \sum_{r \in \RQ} Z^r \mathcal{M}^r\). Crucially, if \(\Q\) possesses a non-linear gauge symmetry, the different monomials \(\mathcal{M}^r\) are related to each other and a general choice for the counterterms \(Z^r\) breaks this gauge invariance. Specifically, they need to satisfy the corresponding \emph{Slavnov--Taylor identities}, which are given by the requirement \(S \LQ^\texttt{R} \overset{!}{\simeq}_\text{TD} 0\). These relations between the counterterms are induced by the symmetry of the theory but also the choice of the renormalization condition. The results of this article aim to understand when a theory satisfies such symmetries and, if yes, provide a construction for a gauge invariant renormalization operation therefor.

To this end, we use the \emph{Hopf algebraic renormalization theory} by Connes and Kreimer \cite{Kreimer_Hopf_Algebra,Connes_Kreimer_NG,Broadhurst_Kreimer,Connes_Kreimer_0,Connes_Kreimer_1,Connes_Kreimer_2}. The startpoint for this framework is a Quantum Field Theory \(\Q\), given via a Lagrange density \(\LQ\) and its set of  \emph{one-particle irreducible Feynman graphs} \(\GQ\). From this, a Hopf algebra \(\HQ\) is constructed as follows: The algebra part is given via the power series algebra \(\HQ \coloneq \mathbb{Q}[[\GQ]]\) with multiplication given via disjoint union and unit given via the empty graph. Crucially, the coproduct is constructed to capture the subdivergence structure of Feynman graphs:
\begin{equation}
	\Delta \, : \quad \HQ \to \HQ \otimes_\mathbb{Q} \HQ \, , \quad \Gamma \mapsto \sum_{\gamma \in \DQ{\Gamma}} \gamma \otimes_\mathbb{Q} \Gamma / \gamma \, ,
\end{equation}
where the cograph \(\Gamma / \gamma\) is defined by shrinking the internal edges of \(\gamma\) in \(\Gamma\) to new vertices for each connected component of \(\gamma\) and \(\DQ{\Gamma}\) denotes the set of superficially divergent subgraphs of \(\Gamma\). Then, Feynman rules correspond to \emph{characters} \(\Phi \colon \HQ \to \AQ\), i.e.\ algebra morphisms, from this Hopf algebra to a suitable target algebra \(\AQ\), e.g.\ an algebra of formal integral expressions. In particular, the module of such maps can be turned into a \emph{character group} \(G^{\HQ}_{\AQ}\). Crucially, the renormalization scheme provides then the input to split the target algebra into a regular and a singular part, i.e.\ \(\AQ \cong_\mathscr{R} \AQ^\text{reg} \oplus \AQ^\text{sing}\). This allows to construct the subgroup of renormalized characters and provides unique maps \(\Phi^\text{reg}\) and \(\Phi^\text{sing}\), which are the \emph{renormalized Feynman rules} and the \emph{counterterm map}, respectively. Notably, the \(Z\)-factors are then the image of the counterterm map under the respective combinatorial Green's function. Since this framework is mathematically very rigid, it allows to work out the precise obstructions for multiplicative renormalization. So far, there have been the following attempts to combine this setup with gauge theories: First, Kreimer himself observed that the insertion operator \(B_+\) is a Hochschild one-cocycle if and only if the insertion graph is a specific gauge invariant linear combination of graphs \cite{Kreimer_Anatomy}. Next, van Suijlekom showed that the Ward--Takahashi identities in Quantum Electrodynamics, the Slavnov--Taylor identities in Quantum Yang--Mills theories and general gauge symmetries correspond to Hopf ideals inside the renormalization Hopf algebra \cite{vSuijlekom_QED,vSuijlekom_QCD,vSuijlekom_BV}. Building on these results, the author showed that this correspondence can be constructed for any \emph{quantum gauge symmetries}, theories with multiple coupling constants and a transversal structure as well as for super- and non-renormalizable Quantum Field Theories \cite{Prinz_3}: This makes it also suitable to study (effective) Quantum General Relativity in the sense of \cite{Kreimer_QG1,Prinz_PhD}. A further perspective was constructed building on the parametric representation of Feynman graphs: In addition to the two well-known \emph{Symanzik polynomials} a third graph polynomial was introduced, the so-called \emph{Corolla polynomial} \cite{Kreimer_Yeats,Kreimer_Sars_vSuijlekom,Sars_PhD}: This polynomial allows to construct, for a given connected \(\phi^3\)-theory graph, a so-called \emph{Corolla differential}, which constructs the corresponding Yang--Mills theory integrand. Notably, this approach also generalizes to spontaneously broken gauge theories \cite{Prinz_1}, Quantum Electrodynamics \cite{Golz_PhD} and graph and cycle cohomology \cite{Kreimer_Sars_vSuijlekom,Berghoff_Knispel}. The present work contains similar ideas, but develops them in a slightly different context, which allows to directly connect them to the well-known \emph{BRST symmetry} \cite{Becchi_Rouet_Stora_1,Becchi_Rouet_Stora_2,Becchi_Rouet_Stora_3,Tyutin}. Additionally, in a follow-up project, the precise relation with the more involved and more powerful \emph{BV formalism} will be studied \cite{Batalin_Vilkovisky_1,Batalin_Vilkovisky_2}.

For this purpose, we study and formalize the \emph{cancellation identities} introduced by 't Hooft and Veltman, which they originally used to prove the Slavnov--Taylor identities of Quantum Yang--Mills theories and its renormalizability \cite{tHooft,tHooft_Veltman,Cvitanovic,tHooft_Veltman_Nobel_Prize,Taylor,Slavnov}, building on the earlier works of Ward and Takahashi in the realm of Quantum Electrodynamics \cite{Ward,Takahashi}. More specifically, they are a graphical way to study how the unphysical \emph{longitudinal modes} travel through gauge boson propagators and vertices: Notably, they relate longitudinal modes to the respective ghost propagators and vertices. Using this argument recursively shows that the Green's functions are indeed transversal, provided the external particles are physical, i.e.\ transversal and on-shell. In present article we formalize and generalize them: Despite the intellectual interest, this also prepares the framework for (effective) Quantum General Relativity, which possesses all graviton vertex-valances and --- depending on the chosen definition of the graviton field --- either all ghost vertex-valences or just the three-valent ghost-vertex.

In particular, we discuss the specific situation of Quantum Yang--Mills theory in \exref{exmp:qym}, given by the Lagrange density
\begin{equation}
\begin{split}
	\mathcal{L}_\text{QYM} & \coloneq \mathcal{L}_\text{YM} + \mathcal{L}_\text{GF} + \mathcal{L}_\text{Ghost} \\ & \phantom{:} \equiv - \eta^{\mu \nu} \eta^{\rho \sigma} \delta_{a b} \left ( \frac{1}{4 \mathrm{g}^2} F^a_{\mu \rho} F^b_{\nu \sigma} + \frac{1}{2 \xi} \big ( \partial_\mu A^a_\nu \big ) \big ( \partial_\rho A^b_\sigma \big ) \right ) \dif V_\eta \\
	& \phantom{\coloneq} + \eta^{\mu \nu} \left ( \frac{1}{\xi} \overline{c}_a \left ( \partial_\mu \partial_\nu c^a \right ) + \mathrm{g} \tensor{f}{^a _b _c} \overline{c}_a \left ( \partial_\mu \big ( c^b A^c_\nu \big ) \right ) \right ) \dif V_\eta \, ,
\end{split}
\end{equation}
where \(F^a_{\mu \nu} := \mathrm{g} \big ( \partial_\mu A^a_\nu - \partial_\nu A^a_\mu \big ) - \mathrm{g}^2 \tensor{f}{^a _b _c} A^b_\mu A^c_\nu\) is the local curvature form of the gauge boson \(A^a_\mu\). Furthermore, \(\dif V_\eta := \dif t \wedge \dif x \wedge \dif y \wedge \dif z\) denotes the Minkowskian volume form. Additionally, \(\eta^{\mu \nu} \partial_\mu A^a_\nu \equiv 0\) is the Lorenz gauge fixing functional and \(\xi\) the gauge fixing parameter. Finally, \(c^a\) and \(\overline{c}_a\) are the gauge ghost and gauge antighost, respectively. In addition, we discuss the more involved situation of (effective) Quantum General Relativity in \exref{exmp:qgr}, given by the Lagrange density
\begin{equation}
\begin{split}
	\mathcal{L}_\text{QGR} & \coloneq \mathcal{L}_\text{GR} + \mathcal{L}_\text{GF} + \mathcal{L}_\text{Ghost} \\ & \phantom{:} \equiv - \frac{1}{2 \varkappa^2} \left ( \sqrt{- \Det{g}} R + \frac{1}{2 \zeta} \eta^{\mu \nu} \deDonder^{(1)}_\mu \deDonder^{(1)}_\nu \right ) \dif V_\eta \\ & \phantom{\coloneq} - \frac{1}{2} \eta^{\rho \sigma} \left ( \frac{1}{\zeta} \overline{C}^\mu \left ( \partial_\rho \partial_\sigma C_\mu \right ) + \overline{C}^\mu \left ( \partial_\mu \big ( \tensor{\Gamma}{^\nu _\rho _\sigma} C_\nu \big ) - 2 \partial_\rho \big ( \tensor{\Gamma}{^\nu _\mu _\sigma} C_\nu \big ) \right ) \right ) \dif V_\eta \, ,
\end{split}
\end{equation}
where \(R := g^{\nu \sigma} \tensor{R}{^\mu _\nu _\mu _\sigma}\) is the Ricci scalar (with \(\tensor{R}{^\rho _\sigma _\mu _\nu} := \partial_\mu \tensor{\Gamma}{^\rho _\nu _\sigma} - \partial_\nu \tensor{\Gamma}{^\rho _\mu _\sigma} + \tensor{\Gamma}{^\rho _\mu _\lambda} \tensor{\Gamma}{^\lambda _\nu _\sigma} - \tensor{\Gamma}{^\rho _\nu _\lambda} \tensor{\Gamma}{^\lambda _\mu _\sigma}\) the Riemann tensor). Again, \(\dif V_\eta := \dif t \wedge \dif x \wedge \dif y \wedge \dif z\) denotes the Minkowskian volume form, which is related to the Riemannian volume form \(\dif V_g\) via \(\dif V_g \equiv \sqrt{- \Det{g}} \dif V_\eta\). Additionally, \(\deDonder^{(1)}_\mu := \eta^{\rho \sigma} \Gamma_{\mu \rho \sigma} \equiv 0\) is the linearized de Donder gauge fixing functional and \(\zeta\) the gauge fixing parameter. Finally, \(C_\mu\) and \(\overline{C}^\mu\) are the graviton-ghost and graviton-antighost, respectively. We refer to \cite{Prinz_4,Prinz_2,Prinz_8} for further details on the perturbative approach to (effective) Quantum General Relativity using the same conventions.

Explicitly, this article is organized as follows: In \sectionref{sec:conn-hopf-alg-ren}, we start with a review and generalization of the Connes--Kreimer theory for \emph{connected} Feynman graphs, by in particular generalizing the coproduct-identities in \propref{prop:coproduct_cgreensfunctions}. Then, in \sectionref{sec:formalized-cancellation-identities}, we review and formalize the cancellation identities of 't Hooft and Veltman, notably for theories with an arbitrary vertex-valence in \thmref{thm:gen-ci}. Next, in \sectionref{sec:pBRST-Feynman-graph-complex}, we implement these formalized cancellation identities into a Feynman graph complex, which we call \emph{pBRST Feynman graph complex} due to its structural similarities with the BRST complex. In particular, we show that its cohomology groups are generated by connected restricted combinatorial Green's functions in \thmref{thm:pBRST-cohomology}. Additionally, in \sectionref{sec:derived-renormalization-theory} we construct the renormalization Hopf algebras on said cohomology group generators and emphasize the relation to Hopf subalgebras for multiplicative renormalization in \colref{col:derived-renormalization-theory}. Finally, in \sectionref{sec:conclusion}, we close with a summary and propose follow-up projects, namely to study the relation of the present constructions to the BV formalism and apply it to UV completions of (effective) Quantum General Relativity.

\section{Connected Hopf algebraic renormalization} \label{sec:conn-hopf-alg-ren}

In this section, we provide the necessary background of Connes--Kreimer renormalization theory \cite{Kreimer_Hopf_Algebra,Connes_Kreimer_0}, building on the more detailed presentations of \cite[Section 2]{Prinz_3} and \cite[Section 3]{Prinz_PhD}. In particular, since we aim to discuss (generalized) gauge theories, we reformulate the framework in terms of \emph{connected Feynman graphs}, rather than \emph{one-particle irreducible (1PI) Feynman graphs}: This is required from the corresponding \emph{Slavnov--Taylor identities} with their corresponding \emph{cancellation identities}, which we will discuss in the following section.

\enter

\begin{defn}[Combinatorial data of a QFT] \label{defn:combinatorial-data-qft}
	Let \(\Q\) be a Quantum Field Theory given via the Lagrange density \(\LQ\), then we associate the following sets to it:
	\begin{itemize}
		\item The residue set \(\RQ \coloneq \RQO \sqcup \RQI\), where \(\RQO\) is the set of interactions (vertex colorings) and \(\RQO\) the set of particles (edge colorings)
		\item The amplitude set \(\AQ \coloneq \RQ \sqcup \mathcal{Q}_{\Q}\), where \(\RQ\) is the residue set and \(\mathcal{Q}_{\Q}\) the set of pure quantum corrections, i.e.\ particle interactions which are only possible via trees or graphs
		\item The connected Feynman graph set \(\GQc \coloneq \GQt \sqcup \GQl\), where \(\GQt\) is the set of tree Feynman graphs and \(\GQl\) the set of connected Feynman graphs with at least one loop
		\item The space of fields \(\FQ\), each of which is represented via a propagator in \(\RQI\)
	\end{itemize}
	We emphasize that usually the following constructions are based on the set \(\GQ \subset \GQc\) of \emph{one particle irreducible (1PI)} Feynman graphs,\footnote{Feynman graphs that remain connected after removal of any of the internal edges, also called \emph{bridge-free} in the mathematics literature.} but (generalized) gauge theories require us to work with \emph{connected} Feynman graphs instead, cf.\ \sectionref{sec:formalized-cancellation-identities}.
\end{defn}

\enter

\begin{defn}[Properties of Feynman graphs] \label{defn:properties-feynman-graphs}
	Given a Feynman graph \(\Gamma \in \GQc\), we associate the following data to it:
	\begin{itemize}
		\item Its vertex set \(V \left ( \Gamma \right )\), its edge set \(E \left ( \Gamma \right )\) and its set of external half-edges \(E_\text{Ext} \left ( \Gamma \right )\), i.e.\ edges that are connected to only one internal vertex
		\item Its residue \(\res{\Gamma} \in \AQ\), obtained by shrinking all internal edges to a new vertex
		\item Its symmetry factor \(\sym{\Gamma} \coloneq \left \vert \aut{\Gamma} \right \vert \in \mathbb{N}_+\), as the rank of its automorphism group
		\item Its superficial degree of divergence \(\sdd{\Gamma} \in \mathbb{Z}\)
		\item Its set of superficially divergent subgraphs \(\DQ{\Gamma} \subset \bigsqcup \GQc\)
		\item Its loop-grading number \(\operatorname{LoopGrd} \left ( \Gamma \right ) \in \mathbb{N}_0\), its vertex-grading multi-vector \(\vtxgrd{\Gamma} \in \ZvQ\), counting the vertex-types, and its coupling-grading multi-vector \(\cplgrd{\Gamma} \in \ZqQ\), counting the corresponding coupling constants
	\end{itemize}
	We refer to \cite[Section 2]{Prinz_3} for the detailed definitions of the above quantities.
\end{defn}
\enter

\begin{defn}[The renormalization Hopf algebra] \label{defn:renormalization_hopf_algebra}
	Given the situation of \defnsaref{defn:combinatorial-data-qft}{defn:properties-feynman-graphs}, we define the following Hopf algebra \(\HQ \equiv (\VQ, m, \one, \Delta, \coone, S)\), called \emph{renormalization Hopf algebra}:
	\begin{itemize}
		\item The algebra part is given via the power series algebra \(\VQ \coloneq \mathbb{Q} [[ \GQc ]]\),\footnote{Strictly speaking, Connes and Kreimer defined it as the polynomial algebra \(\mathbb{Q}[\GQc]\) and then completed it with respect to the loop-grading, which is equivalent to our definition.} with multiplication \(m \colon \HQ \otimes_\mathbb{Q} \HQ \to \HQ\) via disjoint union and unit \(\one \colon \mathbb{Q} \hookrightarrow \HQ\) via the empty graph
		\item The coproduct is given as follows, using the set of superficially divergent subgraphs \(\DQ{\Gamma}\) from \defnref{defn:properties-feynman-graphs}:
			\begin{equation}
				\Delta \, : \quad \HQ \to \HQ \otimes_\mathbb{Q} \HQ \, , \quad \Gamma \mapsto \sum_{\gamma \in \DQ{\Gamma}} \gamma \otimes_\mathbb{Q} \Gamma / \gamma \, ,
			\end{equation}
			where the cograph \(\Gamma / \gamma\) is defined by shrinking the internal edges of \(\gamma\) in \(\Gamma\) to new vertices for each connected component of \(\gamma\) with corresponding counit \(\coone \colon \HQ \surject \mathbb{Q}\) via the map sending all non-empty graphs to zero and the empty graph to its prefactor
		\item Finally, the antipode is recursively defined via the normalization \(S \left ( \one \right ) \coloneq \one\) and on non-trivial graphs as follows:
			\begin{equation}
				S \, : \quad \HQ \to \HQ \, , \quad \Gamma \mapsto - \Gamma - \sum_{\DQprime{\Gamma}} S \left ( \gamma \right ) \Gamma / \gamma \, ,
			\end{equation}
			where the set \(\DQprime{\Gamma}\) is the reduced set of superficially divergent subgraphs, i.e.\ the set \(\DQ{\Gamma}\) without \(\one\) and \(\Gamma\)
	\end{itemize}
	In the following, we will omit the ground field from the tensor product, i.e.\ set \(\otimes \coloneq \otimes_\mathbb{Q}\). Finally, we remark that especially in the context of quantum gauge theories the above construction can be ill-defined, which requires the notion of an `associated renormalization Hopf algebra' \cite[Subsection 3.3]{Prinz_2}: Basically, it is sometimes required to slightly redefine the set of superficially divergent subgraphs \(\DQ{\Gamma}\) such that the quotient \(\Gamma / \gamma\) is well-defined for all \(\gamma \in \DQ{\Gamma}\).
\end{defn}

\enter

\begin{defn}[Convolution product] \label{defn:convolution_product}
	Let \(\ring\) be a ring, \(A\) a \(\ring\)-algebra with multiplication \(m_A \colon A \otimes_\ring A \to A\) and unit \(\one_A \colon \ring \hookrightarrow A\) and \(C\) a \(\ring\)-coalgebra with coproduct \(\Delta_C \colon C \to C \otimes_\ring C\) and counit \(\coone_C \colon C \surject \ring\). Then we can turn the \(\ring\)-module \(\text{Hom}_{\ring-\mathsf{Mod}} \left ( C , A \right )\) of \(\ring\)-linear maps from \(C\) to \(A\) into a \(\ring\)-algebra as follows: Given \(f, g \in \text{Hom}_{\ring-\mathsf{Mod}} \left ( C , A \right )\), we define the \emph{convolution product} via
	\begin{subequations}
	\begin{align}
		f \star g & \coloneq \mult_A \circ \left ( f \otimes g \right ) \circ \Delta_C
		\intertext{with unit given by}
		\mathbf{1}_\star & \coloneq \one_A \circ \coone_C \, .
		\intertext{This definition extends trivially if \(A\) or \(C\) possesses additionally a bi- or Hopf algebra structure. In particular, if \(C\) possesses an antipode \(S \colon C \to C\), we can complement this monoid to the \emph{character group} \(G^C_A\) via the inversion}
		f^{\star -1} & \coloneq f \circ S \, .
	\end{align}
	\end{subequations}
	The convolution product is commutative, if \(C\) is cocommutative and \(A\) is commutative.
\end{defn}

\enter

\begin{defn}[Projection to divergent graphs] \label{defn:projection_divergent_graphs}
	Given the situation of \defnref{defn:renormalization_hopf_algebra}, we define the projection to superficially divergent graphs via
	\begin{equation}
		\Omega \, : \quad \GQc \to \GQc \, , \quad \Gamma \mapsto \begin{cases} \Gamma & \text{if \(\sdd{\Gamma} \geq 0\)} \\ 0 & \text{else, i.e.\ \(\sdd{\Gamma} < 0\)} \end{cases}
	\end{equation}
	and then extend it additively and multiplicatively to \(\HQ\). Additionally, we also use the shorthand-notation
	\begin{equation}
		\overline{\mathfrak{G}} \coloneq \Omega \left ( \mathfrak{G} \right ) \, .
	\end{equation}
	We remark that this definition is useful for coproduct identities related to Hopf subalgebras for multiplicative renormalization in the context of super- or non-renormalizable Quantum Field Theories, cf.\ \cite[Proposition 4.2]{Prinz_3} and \propref{prop:coproduct_cgreensfunctions}.
\end{defn}

\enter

\begin{defn}[(Restricted) combinatorial 1PI Green's functions] \label{def:combgreenipi}
	Let \(\Q\) be a QFT, \(\AQ\) the set of its amplitudes and \(\GQ\) the set of its Feynman graphs. Given an amplitude \(r \in \AQ\),\footnote{Later, we will oftentimes write \((i,j,k)\) instead of \(r\) to indicate an amplitude with \(i\) external gauge bosons, \(j\) external ghosts and \(k\) external matter particles.} we set
	\begin{align}
		\precombgreen^r & \coloneq \sum_{\substack{\Gamma \in \GQ \\ \res{\Gamma} = r}} \frac{1}{\sym{\Gamma}} \Gamma
	\intertext{and then define the combinatorial Green's function with amplitude \(r\) as the following sum:}
		\combgreen^r & \coloneq \begin{cases} \one + \precombgreen^r & \text{if \(r \in \RQ^{[0]}\)} \\ \one - \precombgreen^r & \text{if \(r \in \RQ^{[1]}\)} \\ \precombgreen^r & \text{else, i.e.\ \(r \in \mathcal{Q}_\Q\)} \end{cases} \label{eqn:combgreen}
	\intertext{Furthermore, we denote the restriction of \(\combgreen^r\) to one of the gradings \(\mathbf{g}\) from \defnref{defn:properties-feynman-graphs} via}
		\rescombgreen^r_\mathbf{g} & \coloneq \eval{\combgreen^r}_{\mathbf{g}} \, .
	\end{align}
\end{defn}

\enter

\begin{defn}[(Restricted) combinatorial connected Green's functions] \label{def:combgreenc}
	Given the situation of \defnref{def:combgreenipi}, we set
	\begin{align}
		\precombgreenc^r & \coloneq \sum_{\substack{\Gamma \in \GQc \\ \res{\Gamma} = r}} \frac{1}{\sym{\Gamma}} \Gamma
	\intertext{and then define the combinatorial Green's function with amplitude \(r\) as the following sum:}
		\combgreenc^r & \coloneq \begin{cases} \one + \precombgreenc^r & \text{if \(r \in \RQ\)} \\ \precombgreenc^r & \text{else, i.e.\ \(r \in \mathcal{Q}_\Q\)} \end{cases} \label{eqn:combgreenc}
	\intertext{Furthermore, we denote the restriction of \(\combgreenc^r\) to one of the gradings \(\mathbf{g}\) from \defnref{defn:properties-feynman-graphs} via}
		\rescombgreenc^r_\mathbf{g} & \coloneq \eval{\combgreenc^r}_{\mathbf{g}} \, .
	\end{align}
\end{defn}

\enter

\begin{lem} \label{lem:combgreenc-vs-combgreenipi}
	Given the situation of \defnsaref{def:combgreenipi}{def:combgreenc} and let \(r \in \AQ\) be an amplitude residue. Then, the 1PI and connected combinatorial Green's functions are respectively related as follows:
	\begin{equation} \label{eqn:combgreenc-combgreenipi}
		\combgreenc^r_\mathbf{g} \equiv \begin{cases} \displaystyle \eval{\left ( \frac{1}{\combgreen^r} \right )}_{\mathbf{g}} \equiv \eval{\sum_{k = 0}^\infty \left ( \precombgreen^r \right )^k}_{\mathbf{g}} & \text{if \(r \in \RQI\)} \\ \\ \displaystyle \eval{\left ( \frac{\combgreen^r}{\prod_{e \in E \left ( r \right )} \combgreen^e} \right )}_{\mathbf{g}} & \text{else, i.e.\ \(r \in \RQO \sqcup \mathcal{Q}_\Q\)} \end{cases}
	\end{equation}
\end{lem}

\begin{proof}
	\eqnref{eqn:combgreenc-combgreenipi} follows immediately by applying the formal geometric series to the respective \(\combgreen^e\)'s.
\end{proof}

\enter

\begin{defn}[(Restricted) combinatorial charges] \label{defn:combcharge}
	Let \(v \in \RQ^{[0]}\) be a vertex residue, then we define its combinatorial charge \(\combcharge^v\) via
	\begin{align}
		\combcharge^v & \coloneq \frac{\combgreen^v}{\prod_{e \in E \left ( v \right )} \sqrt{\combgreen^e}} \, ,
	\intertext{where \(E \left ( v \right )\) denotes the set of all edges attached to the vertex \(v\). Furthermore, we denote the restriction of \(\combcharge^v\) to one of the gradings \(\mathbf{g}\) from \defnref{defn:properties-feynman-graphs} via}
		\combcharge^v_\mathbf{g} & \coloneq \eval{\combcharge^v}_{\mathbf{g}} \, .
	\end{align}
\end{defn}

\enter

\begin{lem} \label{lem:combcharges-combgreenc}
	Given the situation of \defnsaref{def:combgreenc}{defn:combcharge}, the combinatorial charges can be also expressed using connected combinatorial Green's functions as follows:
	\begin{equation}
		\combcharge^v \coloneq \frac{\combgreenc^v}{\prod_{e \in E \left ( v \right )} \sqrt{\combgreenc^e}}
	\end{equation}
\end{lem}

\begin{proof}
	This follows directly from \lemref{lem:combgreenc-vs-combgreenipi}.
\end{proof}

\enter

\begin{defn}[(Restricted) products of combinatorial charges] \label{defn:combinatorial_charges}
	Let \(\mathbf{v} \in \ZvQ\) be a multi-index of vertex residues. Then we define the product of combinatorial charges associated to \(\mathbf{v}\) via
	\begin{equation} \label{eqn:combinatorial_charges}
		\combcharge^\mathbf{v} \coloneq \prod_{k = 1}^{\mathfrak{v}_\Q} \left ( \combcharge^{v_k} \right )^{(\mathbf{v})_k} \, ,
	\end{equation}
	where \((\mathbf{v})_k\) denotes the \(k\)-th entry of \(\mathbf{v}\). In particular, given a vertex residue \(v \in \RQO\) and a natural number \(n \in \mathbb{N}_+\), we define the exponentiation of the combinatorial charge \(\combcharge^v\) by \(n\) via
	\begin{equation}
		\combcharge^{nv} \coloneq \left ( \combcharge^{v} \right )^n \, .
	\end{equation}
	Furthermore, we denote the restriction of \(\combcharge^\mathbf{v}\) to one of the gradings \(\mathbf{g}\) from \defnref{defn:properties-feynman-graphs} via
	\begin{equation} \label{eqn:restricted_products_combinatorial_charges}
		\combcharge^\mathbf{v}_\mathbf{g} \coloneq \eval{\left ( \prod_{k = 1}^{\mathfrak{v}_\Q} \left ( \combcharge^{v_k} \right )^{(\mathbf{v})_k} \right )}_\mathbf{g} \, .
	\end{equation}
\end{defn}

\enter

\begin{defn}[Target algebra and (renormalized) Feynman rules] \label{defn:target-algebra_feynman-rules}
	Let \(\Q\) be a QFT, \(\HQ\) its (associated) renormalization Hopf algebra and an algebra of formal integral expressions (cf.\ \cite[Definition 2.35]{Prinz_3} for the proper definition)
	\begin{equation}
		\Phi \, : \quad \HQ \to \EQ \, , \quad \Gamma \mapsto (D_\Gamma, I_\Gamma) \, ,
	\end{equation}
	where \((D_\Gamma, I_\Gamma)\) is the pair of a Feynman differential form together with a domain related to the graph \(\Gamma\).\footnote{This generality allows to address different representations, e.g.\ momentum, parametric and position space.} 	Moreover, we introduce a renormalization scheme as a linear endomorphism
	\begin{equation} \label{eqn:renormalization_scheme}
		\mathscr{R} \, : \quad \EQ^\varepsilon \surject \EQ^\varepsilon \, , \quad \left ( D,I_\mathscr{E} \left ( \varepsilon \right ) \right ) \mapsto \begin{cases} (D,0_D) & \text{if \(\left ( D,I_\mathscr{E} \left ( \varepsilon \right ) \right ) \in \operatorname{Ker} \left ( \mathscr{R} \right )\)} \\ \left ( D,I_{\mathscr{E}, \mathscr{R}} \left ( \varepsilon \right ) \right ) & \text{else} \end{cases}
	\end{equation}
	satisfying \(\mathscr{R}^2 = \mathscr{R}\), where \(0_D\) is the zero differential form on \(D\). Additionally, to ensure locality of the counterterm, \(\mathscr{R}\) needs to be a Rota--Baxter operator of weight \(\lambda = -1\), i.e.\ fulfill
	\begin{equation}
		\mu \circ \left ( \mathscr{R} \otimes \mathscr{R} \right ) + \mathscr{R} \circ \mu = \mathscr{R} \circ \mu \circ \left ( \mathscr{R} \otimes \id + \id \otimes \mathscr{R} \right ) \, ,
	\end{equation}
	where \(\mu\) denotes the multiplication on \(\EQ^\varepsilon\) (and by abuse of notation also on \(\EQ^\varepsilon_-\) via restriction). In particular, \((\EQ^\varepsilon, \mathscr{R})\) is a Rota--Baxter algebra of weight \(\lambda = -1\) and \(\mathscr{R}\) induces the splitting
	\begin{equation}
		\EQ^\varepsilon \cong \EQ^\varepsilon_+ \oplus \EQ^\varepsilon_-
	\end{equation}
	with \(\EQ^\varepsilon_+ \coloneq \operatorname{Ker} \left ( \mathscr{R} \right )\) and \(\EQ^\varepsilon_- \coloneq \operatorname{Im} \left ( \mathscr{R} \right )\). Then we can introduce the counterterm map \(\countertermsymbol\), sometimes also called `twisted antipode', recursively via the normalization \(\counterterm{\one} \in \mathbf{1}_{\EQ^\varepsilon}\) and
	\begin{equation}
		\countertermsymbol \, : \quad \operatorname{Aug} \left ( \HQ \right ) \to \EQ^\varepsilon_- \, , \quad \Gamma \mapsto - \renscheme{\countertermsymbol \star \left ( \regFR \circ \mathscr{A} \right )} \left ( \Gamma \right )
	\end{equation}
	else, where \(\mathscr{A} \colon \HQ \surject \operatorname{Aug} \left ( \HQ \right )\) is the projector onto the augmentation ideal. Next we define renormalized Feynman rules via
	\begin{equation}
		\Phi_\mathscr{R} \, : \quad \HQ \to \EQ^\varepsilon_+ \, , \quad \Gamma \mapsto \left ( \countertermsymbol \star \Phi \right ) \left ( \Gamma \right ) \, ,
	\end{equation}
	where the corresponding formal Feynman integral expression is well-defined for a suitable choice of the renormalization scheme \(\mathscr{R}\). We remark that the renormalized Feynman rules \(\Phi_\mathscr{R}\) and the counterterm map \(\countertermsymbol\) correspond to the algebraic Birkhoff decomposition of the Feynman rules \(\Phi\) with respect to the renormalization scheme \(\mathscr{R}\), as was first observed in \cite{Connes_Kreimer_0} and is e.g.\ reviewed in \cite{Guo,Panzer}.
\end{defn}

\enter

\begin{defn}[Hopf subalgebras for multiplicative renormalization] \label{defn:hopf_subalgebras_renormalization_hopf_algebra}
	Let \(\Q\) be a QFT, \(\HQ\) its (associated) renormalization Hopf algebra \(\rescombgreen^r_\mathbf{G}\) a restricted combinatorial Green's function and \(\rescombgreenc^r_\mathbf{G}\) a restricted connected combinatorial Green's function, where \(\mathbf{G}\) and \(\mathbf{g}\) denote one of the gradings from \defnref{defn:properties-feynman-graphs}. We are interested in Hopf subalgebras which correspond to multiplicative renormalization, i.e.\ Hopf subalgebras of \(\HQ\) such that the coproduct factors over restricted combinatorial Green's functions as follows:
	\begin{subequations} \label{eqn:hopf_subalgebras_multi-index}
	\begin{align}
		\Delta \left ( \rescombgreen^r_{\mathbf{G}} \right ) & = \sum_{\mathbf{g}} \mathfrak{P}_{\mathbf{g}} \left ( \rescombgreen^r_{\mathbf{G}} \right ) \otimes \rescombgreen^r_{\mathbf{G} - \mathbf{g}}
		\intertext{and}
		\Delta \left ( \rescombgreenc^r_{\mathbf{G}} \right ) & = \sum_{\mathbf{g}} \widehat{\mathfrak{P}}_{\mathbf{g}} \left ( \rescombgreenc^r_{\mathbf{G}} \right ) \otimes \rescombgreenc^r_{\mathbf{G} - \mathbf{g}} \, ,
	\end{align}
	\end{subequations}
	where \(\mathfrak{P}_{\mathbf{g}} \left ( \rescombgreen^r_{\mathbf{G}} \right ) \in \HQ\) and \(\widehat{\mathfrak{P}}_{\mathbf{g}} \left ( \rescombgreenc^r_{\mathbf{G}} \right ) \in \HQ\) are polynomials in graphs such that each summand has multi-index \(\mathbf{g}\).\footnote{There exist closed expressions for the polynomial \(\mathfrak{P}_{\mathbf{g}} \left ( \rescombgreen^r_{\mathbf{G}} \right )\), cf.\ \cite[Lemma 4.6]{Yeats_PhD}, \cite[Proposition 16]{vSuijlekom_QCD}, \cite[Theorem 1]{Borinsky_Feyngen} and \cite[Proposition 4.2]{Prinz_3}, and also for the polynomial \(\widehat{\mathfrak{P}}_{\mathbf{g}} \left ( \rescombgreenc^r_{\mathbf{G}} \right )\), cf.\ \propref{prop:coproduct_cgreensfunctions}.}
\end{defn}

\enter

\begin{rem} \label{rem:hopf_subalgebras_renormalization_hopf_algebra}
	Given the situation of \defnref{defn:hopf_subalgebras_renormalization_hopf_algebra} the Hopf subalgebras correspond in the following way to multiplicative renormalization: Since the counterterm map is multiplicative, we can calculate the \(Z\)-factor for a given residue \(r \in \RQ\) via
	\begin{equation}
		Z^r_{\mathscr{E}, \mathscr{R}} \left ( \varepsilon \right ) \coloneq \counterterm{\combgreen^r} \, .
	\end{equation}
	More details in this direction can be found in \cite{Prinz_3,Panzer,vSuijlekom_Multiplicative} (with potentially different notation). Additionally, we remark that the existence of the Hopf subalgebras from \defnref{defn:hopf_subalgebras_renormalization_hopf_algebra} depends crucially on the grading \(\mathbf{g}\). In particular, for the loop-grading these Hopf subalgebras exist if and only if \(\Q\) has only one fundamental interaction. Furthermore, they exist for the coupling-grading if and only if \(\Q\) has for each fundamental interaction a different coupling constant. Finally, they always exist for the vertex-grading, cf.\ \cite[Proposition 4.2]{Prinz_3} and \propref{prop:coproduct_cgreensfunctions}.
\end{rem}

\enter

\begin{prop}[Coproduct identities for connected (divergent/restricted) combinatorial Green's functions\footnote{This is the version for combinatorial Green's functions based on \emph{connected graphs}, cf.\ the \emph{1PI versions} here: \cite[Lemma 4.6]{Yeats_PhD}, \cite[Proposition 16]{vSuijlekom_QCD}, \cite[Theorem 1]{Borinsky_Feyngen} and \cite[Proposition 4.2]{Prinz_3}.}] \label{prop:coproduct_cgreensfunctions}
	Let \(\Q\) be a QFT, \(\HQ\) its (associated) renormalization Hopf algebra, \(r \in \RQ\) a residue and \(\mathbf{V} \in \ZvQ\) a vertex-grading multi-index. Then the following identities hold:
	{
	\allowdisplaybreaks
	\begin{subequations}
	\begin{align}
		\Delta \left ( \combgreenc^r \right ) & = \sum_{\mathbf{v} \in \ZvQ} \overline{\combgreenc}^r \overline{\combcharge}^\mathbf{v} \otimes \rescombgreenc^r_{\mathbf{v}} \, , \label{eqn:coproduct_greensfunctions} \\
		\Delta \left ( \rescombgreenc^r_{\mathbf{V}} \right ) & = \sum_{\mathbf{v} \in \ZvQ} \eval{\left ( \overline{\combgreenc}^r \overline{\combcharge}^\mathbf{v} \right )}_{\mathbf{V} - \mathbf{v}} \otimes \rescombgreenc^r_{\mathbf{v}} \, , \label{eqn:coproduct_greensfunctions_restricted}
		\intertext{and, provided that \(\Q\) is cograph-divergent and \(\overline{\rescombgreenc}^r_{\mathbf{V}} \neq 0\),}
		\Delta \big ( \overline{\rescombgreenc}^r_{\mathbf{V}} \big ) & = \sum_{\mathbf{v} \in \ZvQ} \eval{\left ( \overline{\combgreenc}^r \overline{\combcharge}^\mathbf{v} \right )}_{\mathbf{V} - \mathbf{v}} \otimes \overline{\rescombgreenc}^r_{\mathbf{v}} \, . \label{eqn:coproduct_greensfunctions_restricted_divergent}
	\end{align}
	\end{subequations}
	}%
\end{prop}

\begin{proof}
	Using \cite[Lemma 2.27]{Prinz_3}, we know that Feynman graphs with specified residue \(r\) and vertex grading \(\mathbf{v}\) have well-defined vertex and edge sets \(V \left ( r, \mathbf{v} \right )\) and \(E \left ( r, \mathbf{v} \right )\), respectively. Thus, the corresponding corrections (i.e.\ left-hand sides of the tensor product in \eqnref{eqn:coproduct_greensfunctions}) read
	\begin{equation}
	\begin{split}
		\mathcal{I}^{r, \mathbf{v}} & \coloneq \left ( \prod_{v \in V \left ( r, \mathbf{v} \right )} \overline{\combgreenc}^v \right ) \left ( \prod_{e \in E \left ( r, \mathbf{v} \right )} \overline{\combgreenc}^e \right ) \\
		& \phantom{:} \equiv \overline{\combgreenc}^r \overline{\combcharge}^\mathbf{v} \vphantom{\frac{1}{1}} \, ,
	\end{split}
	\end{equation}
	where we have used \lemref{lem:combcharges-combgreenc}, in the second equality.
\end{proof}

\section{Formalized cancellation identities} \label{sec:formalized-cancellation-identities}

Cancellation identities were introduced by 't Hooft as a graphical way to prove the transversality of the perturbative expansion and thus the validity of the Slavnov--Taylor identities in Quantum Yang--Mills theory \cite{tHooft}, cf.\ \cite{tHooft_Veltman,Cvitanovic,Kissler_Kreimer,Gracey_Kissler_Kreimer,Kissler,Prinz_7}. Specifically, gauge boson propagators possess a tensor structure which decomposes into physical and unphysical modes, the latter corresponding to linearized gauge transformations. Then, the unphysical modes of the propagator --- which are called gaugeons --- are related to the ghost propagator. Next, these identities relate linearized gauge transformations of vertices to the corresponding ghost vertices and the higher-valent vertices. In this section, we first list reasonable axioms for (generalized) Quantum Gauge Theories and then present the most general form of cancellation identities ensuring transversality. In particular, we study which terms cancel in pairs and thus obtain dependencies between gauge boson vertices, ghost vertices and their higher-valent vertices. Finally, we exemplify the general theory with Quantum Yang--Mills theory and (effective) Quantum General Relativity.

\enter

\begin{ter}[(Generalized) Gauge Theory]
	Let \(\Q\) be a Quantum Field Theory, given via a Lagrange density \(\LQ\) based on its space of fields \(\FQ\). We call \(\Q\) a (generalized) Gauge Theory, if it possesses a BRST symmetry on the space of fields, i.e.\ there exists a nilpotent operator \(S\) such that for a field \(\varphi \in \FQ\) we have \(S \varphi \in \operatorname{Sym} \bigl ( J^k \FQ \bigr )\), i.e.\ it is a series in fields and a finite number of derivatives thereof.
\end{ter}

\enter

\begin{defn}[Longitudinal and ghost projection]
	Let \(\Q\) be a (generalized) Gauge Theory with residue set \(\RQ\) and space of fields \(\FQ\). We introduce the following two operations: On gauge boson edges, we denote the ghost projection via \(\boldsymbol{\vartriangleright}\) and on gauge boson vertices, we denote the longitudinal projection via \(\boldsymbol{\blacktriangleright}\). The ghost projection corresponds to a linearized gauge transformation and the longitudinal projection to a linearized gauge transformation.
\end{defn}

\enter

\begin{thm}[Generalized cancellation identities\footnote{The Feynman graphs are drawn with \texttt{FeynGame} \cite{Harlander_Klein_Lipp,Harlander_Klein_Schaaf,Bundgen_Harlander_Klein_Schaaf}.}] \label{thm:gen-ci}
	Let \(\rescombgreenc^{i, j, k}_\mathbf{c}\) be the connected Green's function with \(i\) external gauge edges, \(j\) external ghost edges and \(k\) external matter edges in coupling-grading degree \(\mathbf{c}\), where we consider the external edges numbered as follows: Edges \(1, \dots, i\) are gauge bosons, edges \((i+1), \dots, (i+j)\) are ghosts and edges \((i+j+1), \dots, (i+j+k)\) are matter particles.\footnote{If the theory possesses several gauge particles, like (effective) Quantum General Relativity coupled to the Standard Model, then only the gauge bosons for the considered longitudinal contraction \(\boldsymbol{\blacktriangleright}\) and the gaugeon contraction \(\boldsymbol{\vartriangleright}\) count here as gauge edges while the other ones behave as matter edges. Furthermore, we emphasize that we draw matter edges without orientation: Thus, if the corresponding matter field possesses an orientation the corresponding graph is understood as the sum over all such compatible edge orientations.} Let furthermore
	\begin{subequations}
	\begin{align}
		\RQ & \coloneq \RQO \sqcup \RQI
		\intertext{be the corresponding residue set, with}
		\RQO & \coloneq \set{\graph{2cm}{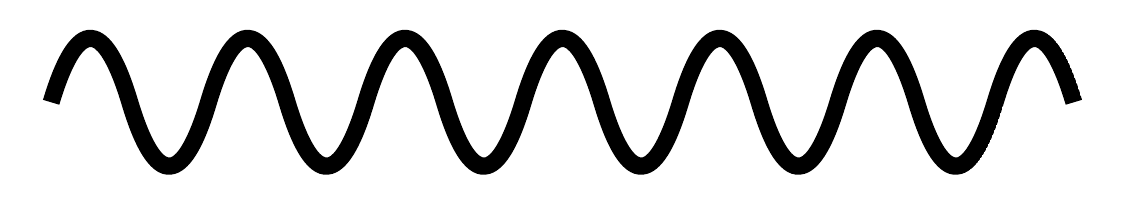} \, , \; \graph{2cm}{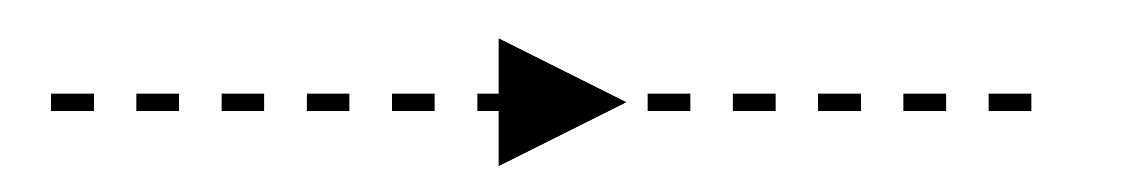} \, , \; \graph{2cm}{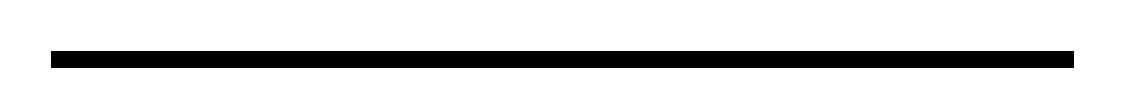}}
		\intertext{the propagator residues, and}
		\RQI & \coloneq \set{\alpha_{i+1} \graph{2.5cm}{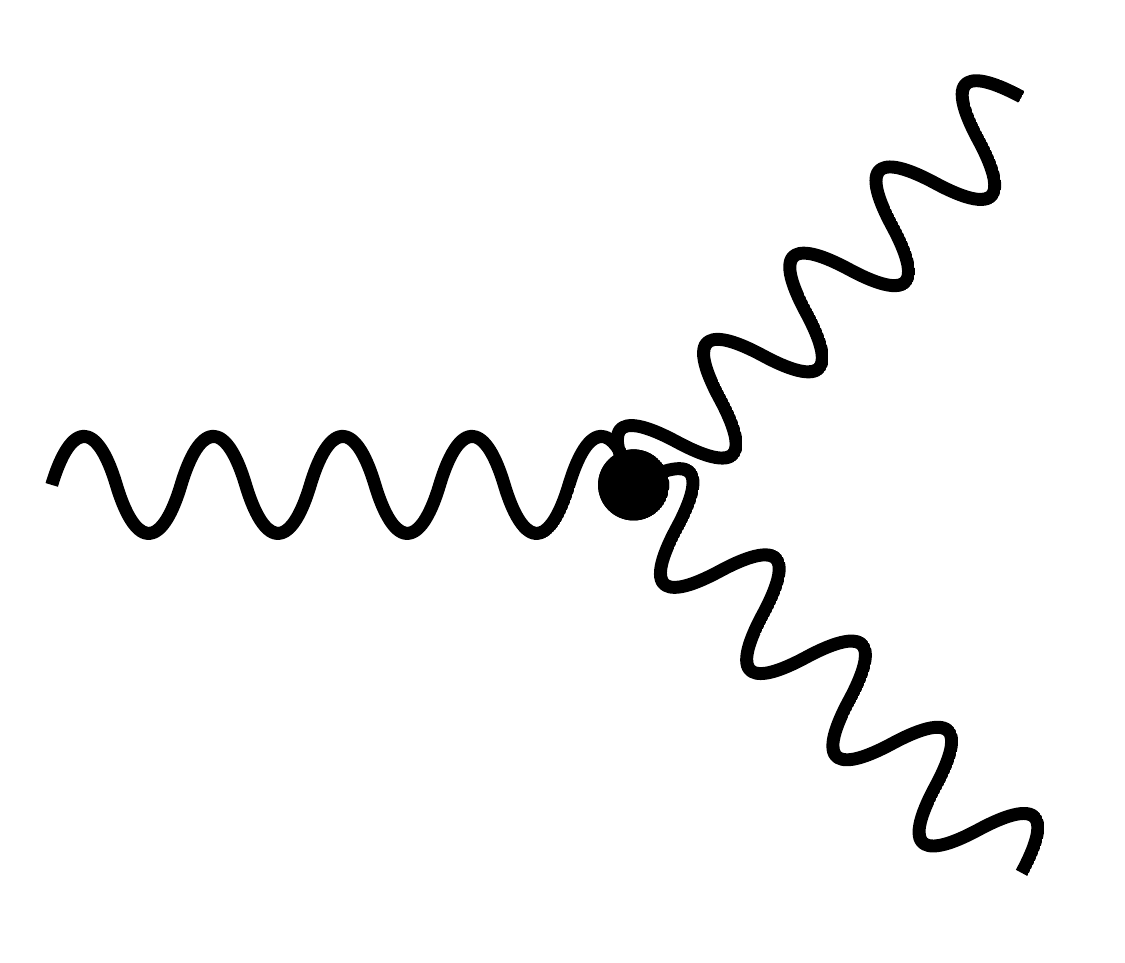} {\mkern-21mu \scriptstyle i \, \Biggr )} \, , \; \beta_{i+1} \graph{2.5cm}{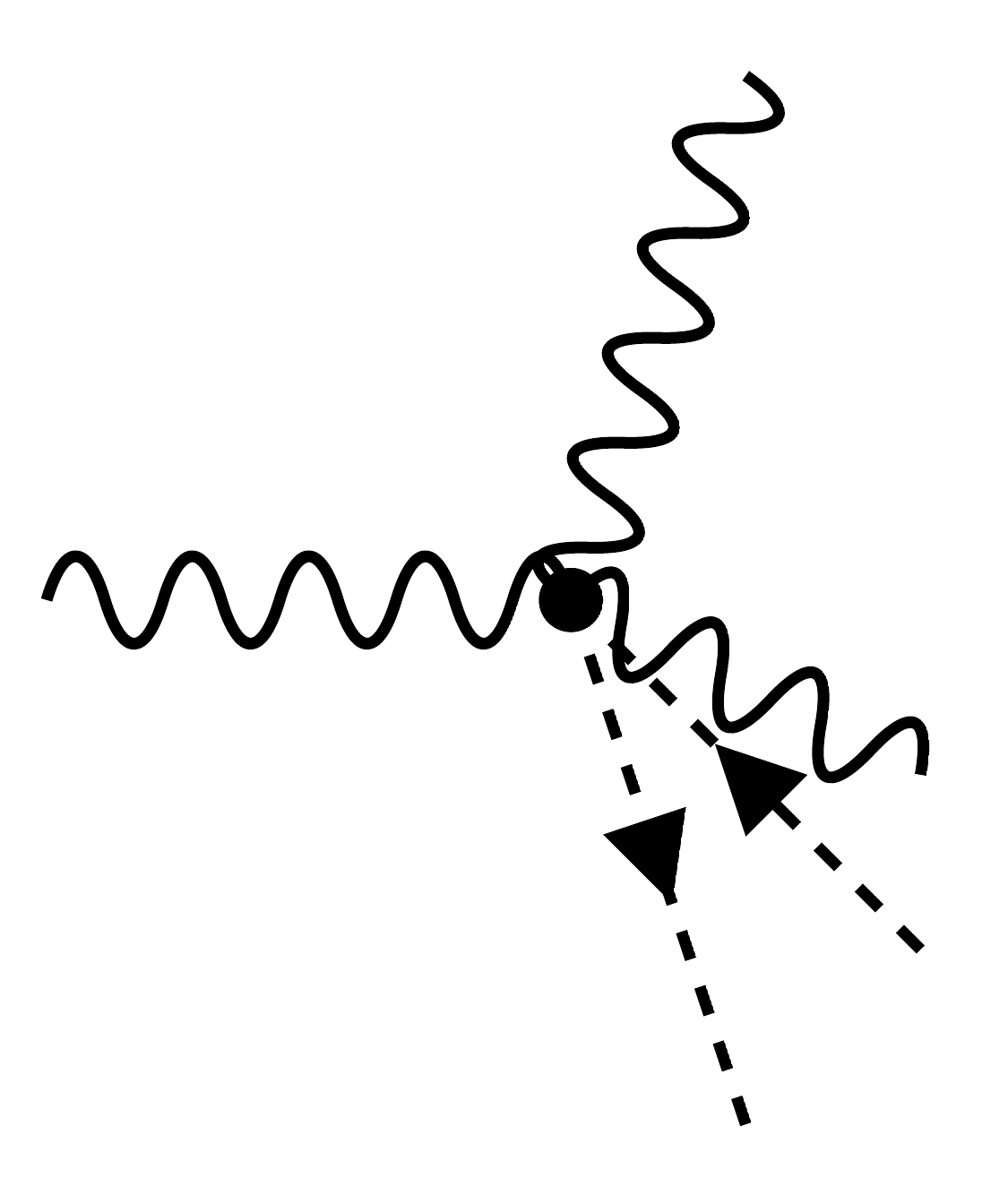} {\mkern-27mu \raisebox{5mm}{\(\scriptstyle i \, \Biggr )\)}} \, , \; \gamma_{i+1,k}\graph{3cm}{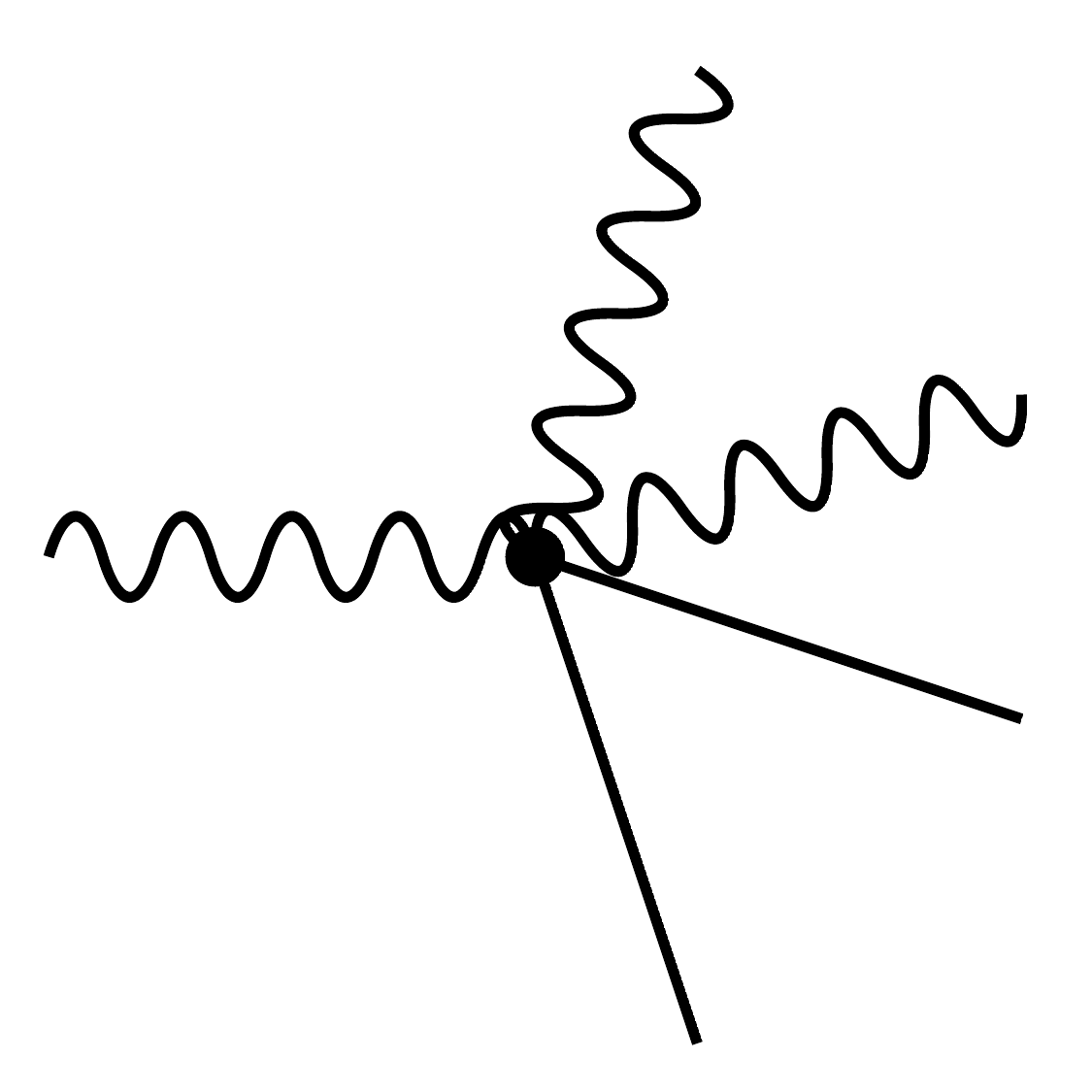}^{\mkern-50mu \raisebox{-6mm}{\(\scriptstyle i \Bigr )\)}}_{\mkern-50mu \raisebox{6mm}{\(\scriptstyle k \Bigr )\)}}}_{i, k \in \mathbb{N}_0}
	\end{align}
	\end{subequations}
	the vertex residues, where \(\alpha_i, \beta_i, \gamma_{i,k} \in \set{0,1}\) are theory-dependent constants determining which vertex-valences are present. Then, the following contraction identities for any given \(i, j, k \in \mathbb{N}_0\) and \(\mathbf{c} \in \ZqQ\)
	\begin{equation}
		\boldsymbol{\blacktriangleright}_1 \cdot \rescombgreenc^{i, j, k}_\mathbf{c} \simeq_\textup{CI} \sum_{l=2}^i \rescombgreenc^{i-2, j+2, k}_\mathbf{c} \cdot \boldsymbol{\vartriangleright}_l + \sum_{l=2}^{i+j+k} \rescombgreenc^{i-1, j+1, k}_\mathbf{c} \cdot \boldsymbol{\vert}_l \, ,
	\end{equation}
	i.e.\ either the longitudinal mode travels through the Green's function and exits via any external gauge edge or it causes an on-shell cancellation on any external edge, are given by the following propagator and vertex identities:
	{\allowdisplaybreaks
	\begin{subequations} \label{eqns:cancellation-identitites}
	\begin{align}
		\gfproj \mkern-6.63mu \graph{2cm}{gauge-propagator} & \simeq_\textup{CI} \graph{2cm}{ghost-propagator} \mkern-7.19mu \gtproj \, ,
		\intertext{i.e.\ the ghost projection of a gauge boson propagator is its ghost propagator with a longitudinal projection at the end,}
		\begin{split}
			\gtproj \mkern-8.25mu \graph{2.5cm}{gauge-vertex} {\mkern-21mu \scriptstyle i \, \Biggr )} & \simeq_\textup{CI} \sum_{l = 2}^i \left ( \beta_{i-1} \graph{2.5cm}{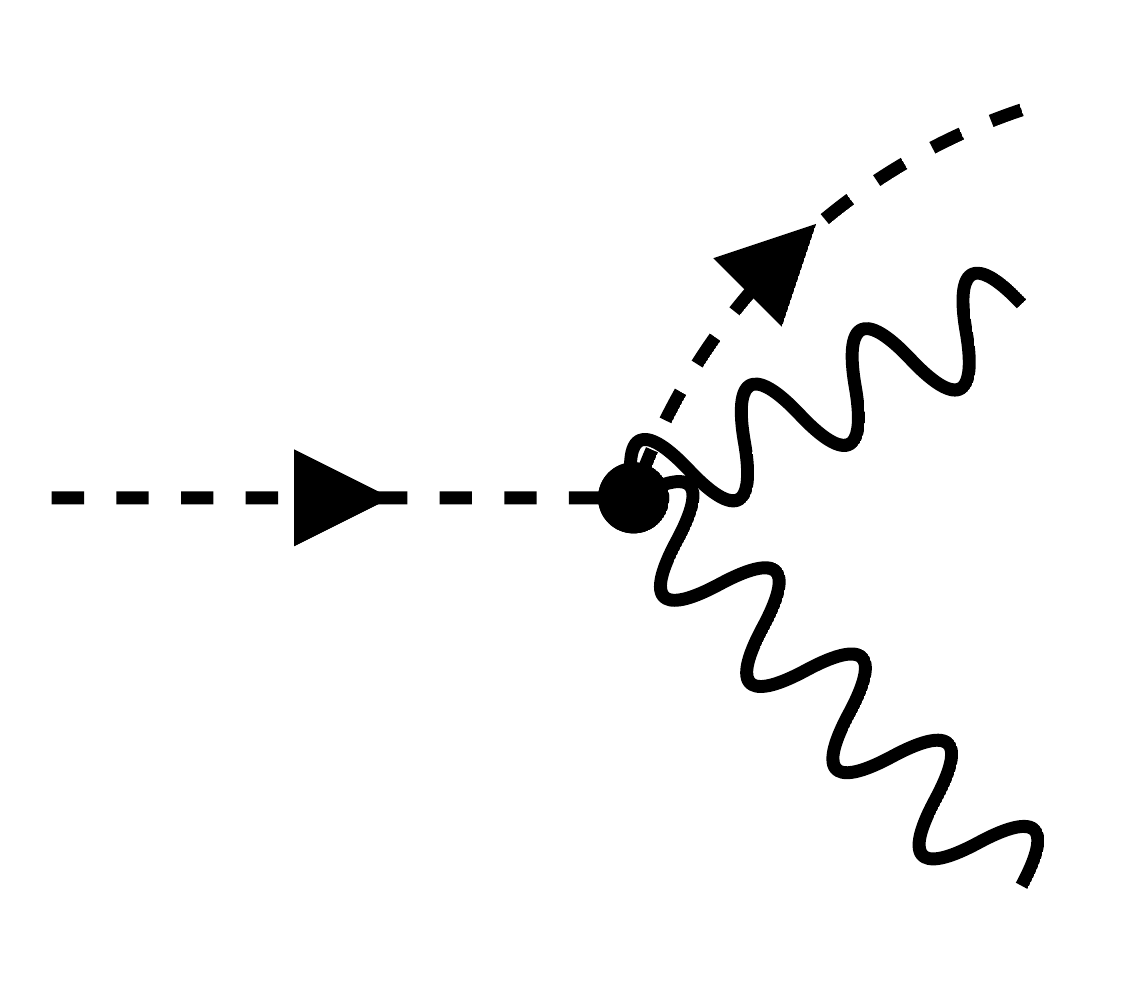}^{\raisebox{-2.25mm}{\(\mkern-12.5mu \gfproj \mkern6mu \scriptstyle l\)}} \raisebox{-1.5mm}{\(\mkern-45mu \scriptstyle i - 1 \biggr )\)} \; + \alpha_{i+1} \graph{2.5cm}{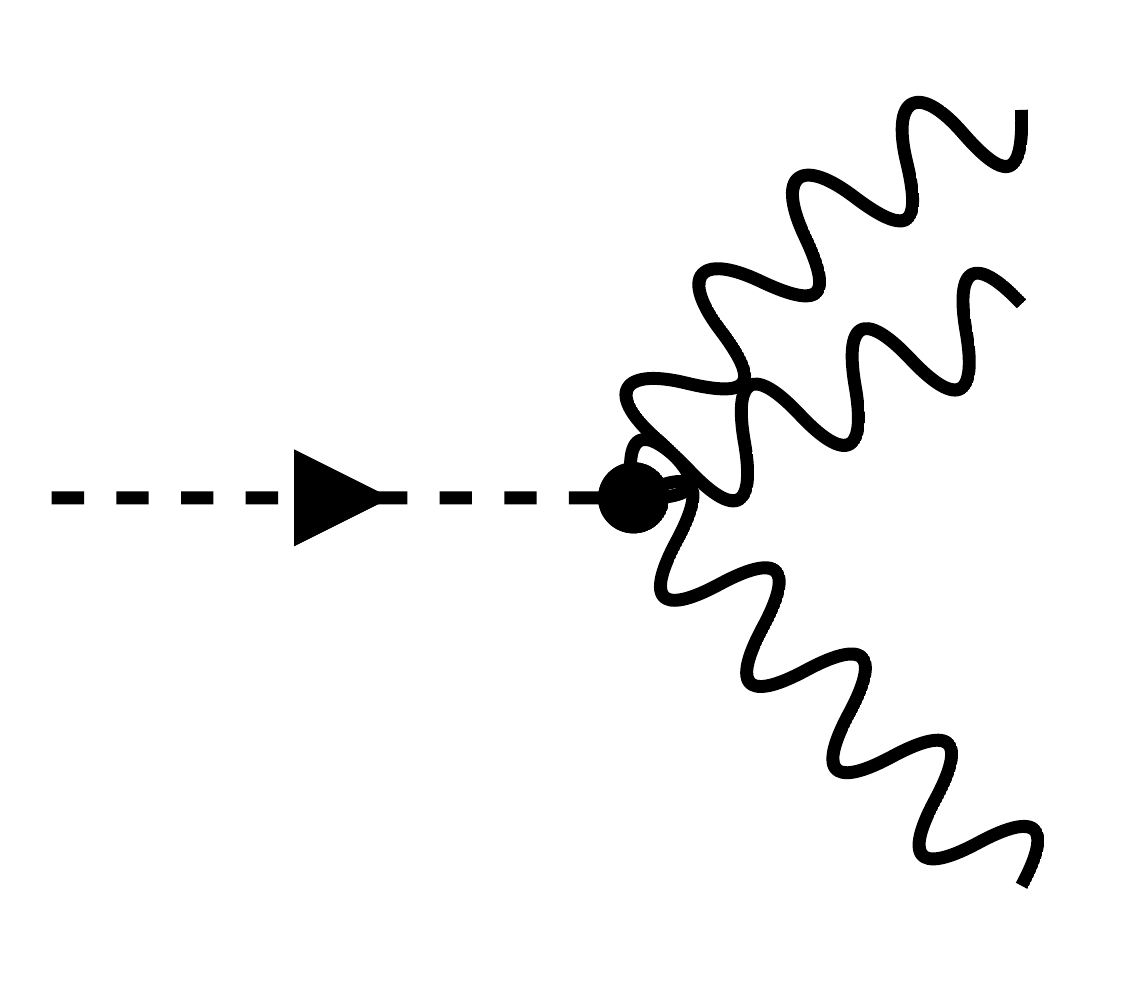}^{\raisebox{-2.25mm}{\(\mkern-14.5mu \oscan \mkern6mu \scriptstyle l \mkern2mu\)}}\raisebox{-1.5mm}{\(\mkern-31mu \scriptstyle i - 1 \biggr )\)} \right ) \\
			& \phantom{=} + \mkern-5mu \sum_{m+n = i} \left ( \alpha_{m+2} \, \alpha_{n+1} \graph{4cm}{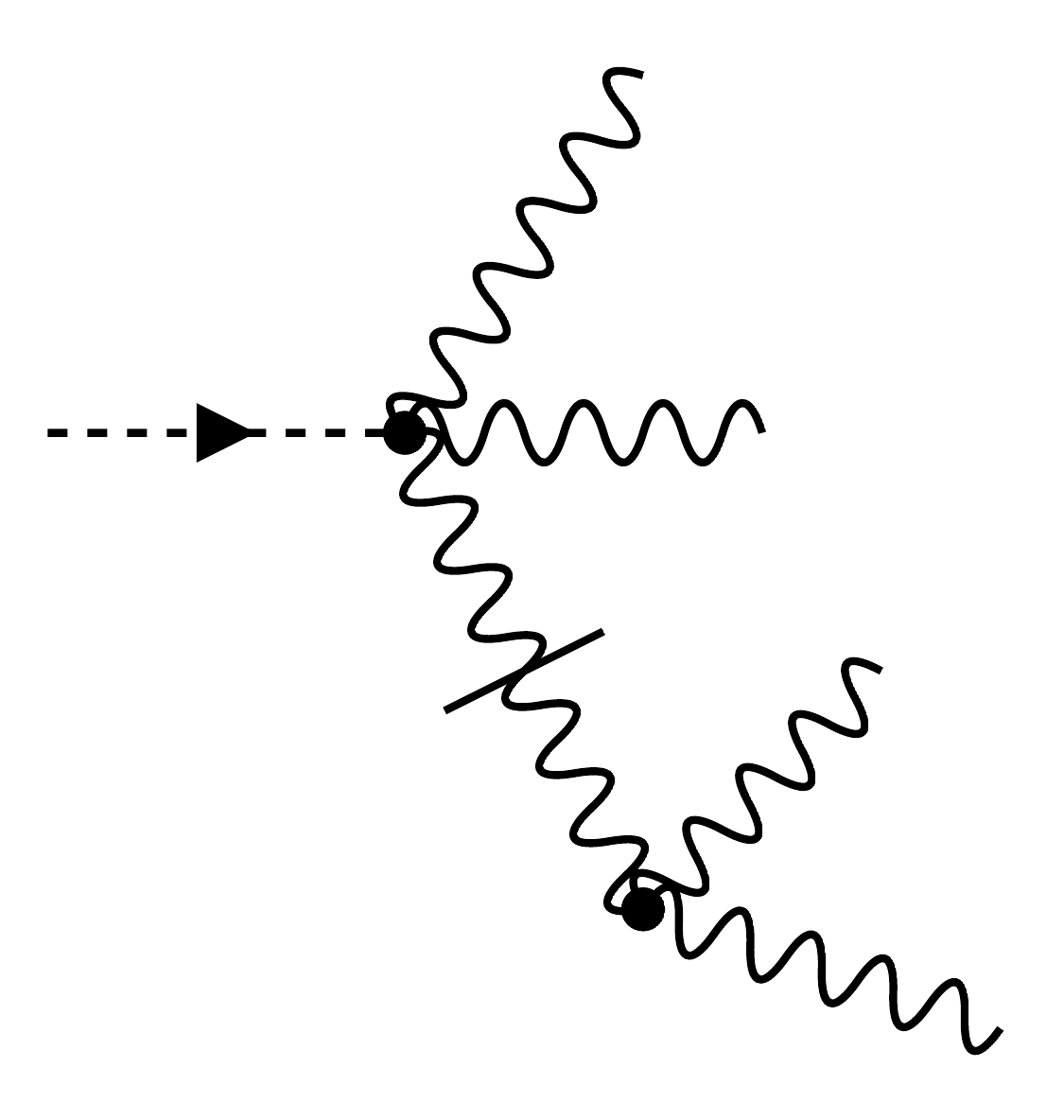}^{\mkern-100mu \raisebox{-9mm}{\(\scriptstyle m \biggr )\)}}_{\mkern-45mu \raisebox{10mm}{\(\scriptstyle n \biggr )\)}} \right ) \, ,
		\end{split}
		\intertext{i.e.\ the longitudinal projection of a gauge boson vertex is its ghost vertex with a ghost projection as well as the gauge boson vertex with an external on-shell cancellation and all internal on-shell expansions of the vertex,}
		\begin{split}
			\gtproj \mkern-8.25mu \graph{2.5cm}{gauge-ghost-vertex} {\mkern-27mu \raisebox{5mm}{\(\scriptstyle i \, \Biggr )\)}} & \simeq_\textup{CI} \sum_{l = 2}^i \left ( \beta_{i+2} \graph{2.5cm}{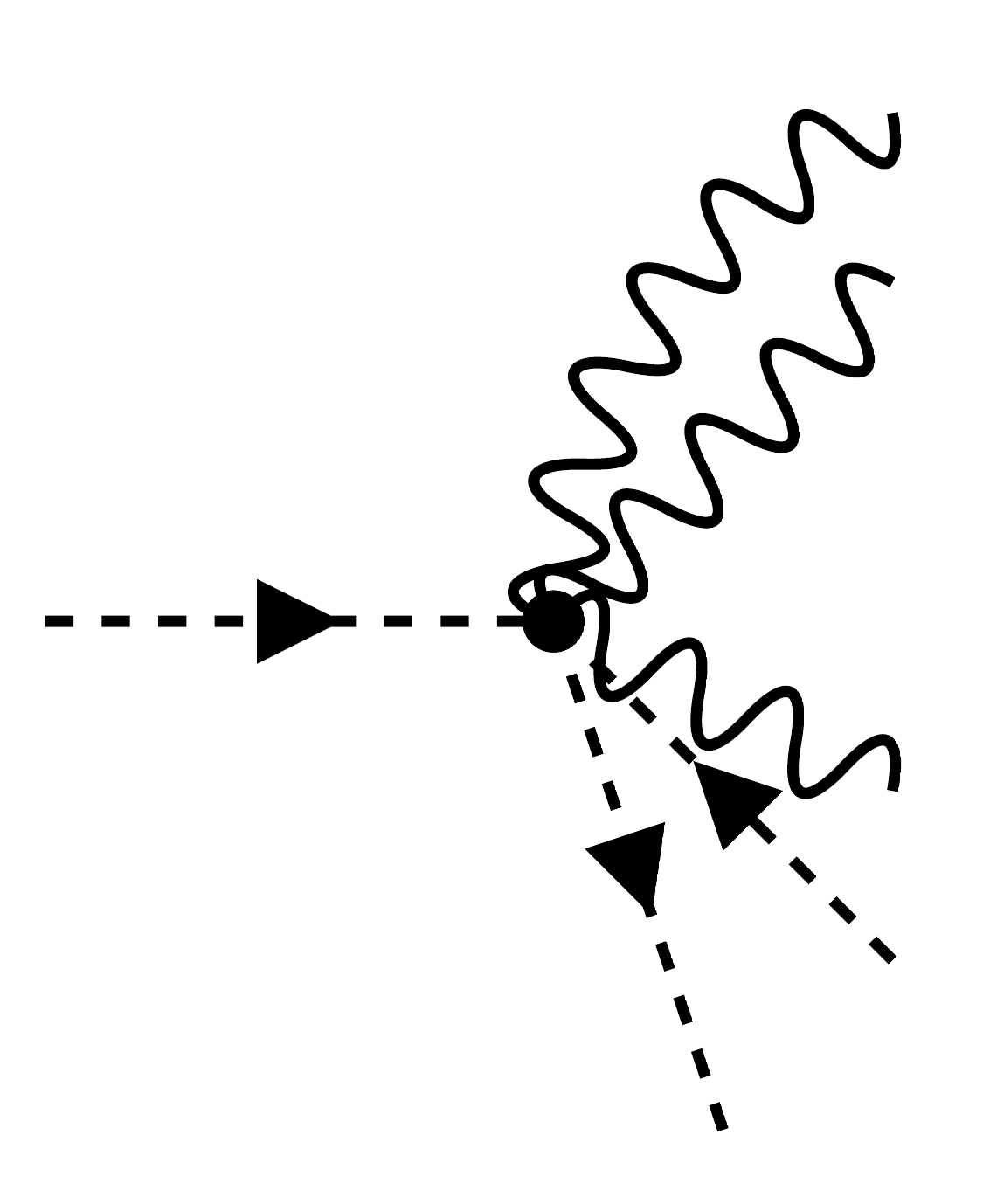}^{\raisebox{-3mm}{\(\mkern-14.25mu \oscan \mkern6mu \scriptstyle l \)}}\mkern-31.5mu \raisebox{1.5mm}{\(\scriptstyle i - 1 \biggr )\)} \right )  \! + \beta_{i+2} \! \left (\mkern-6mu \graph{2.5cm}{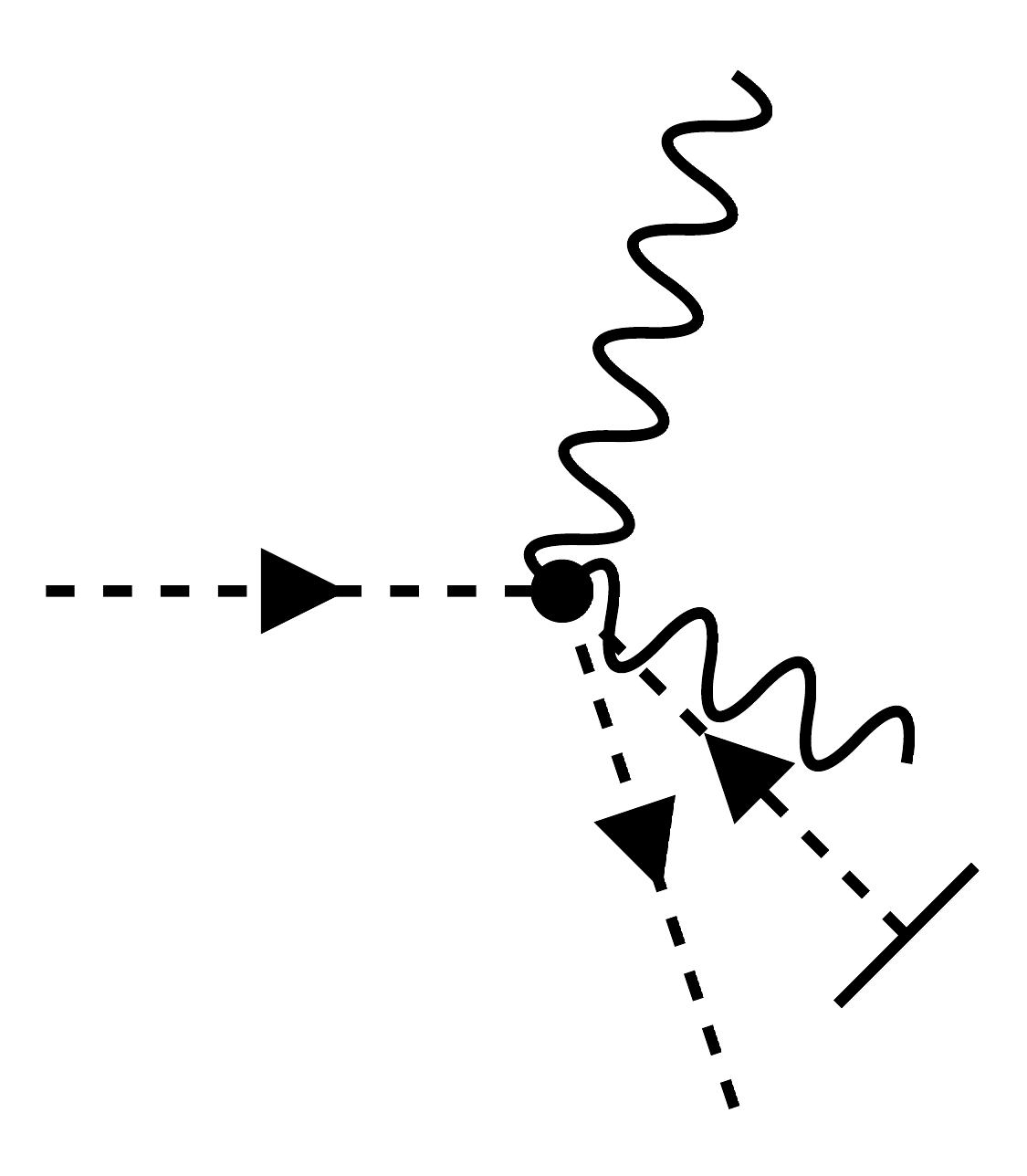} \mkern-12mu + \mkern-6mu \graph{2.5cm}{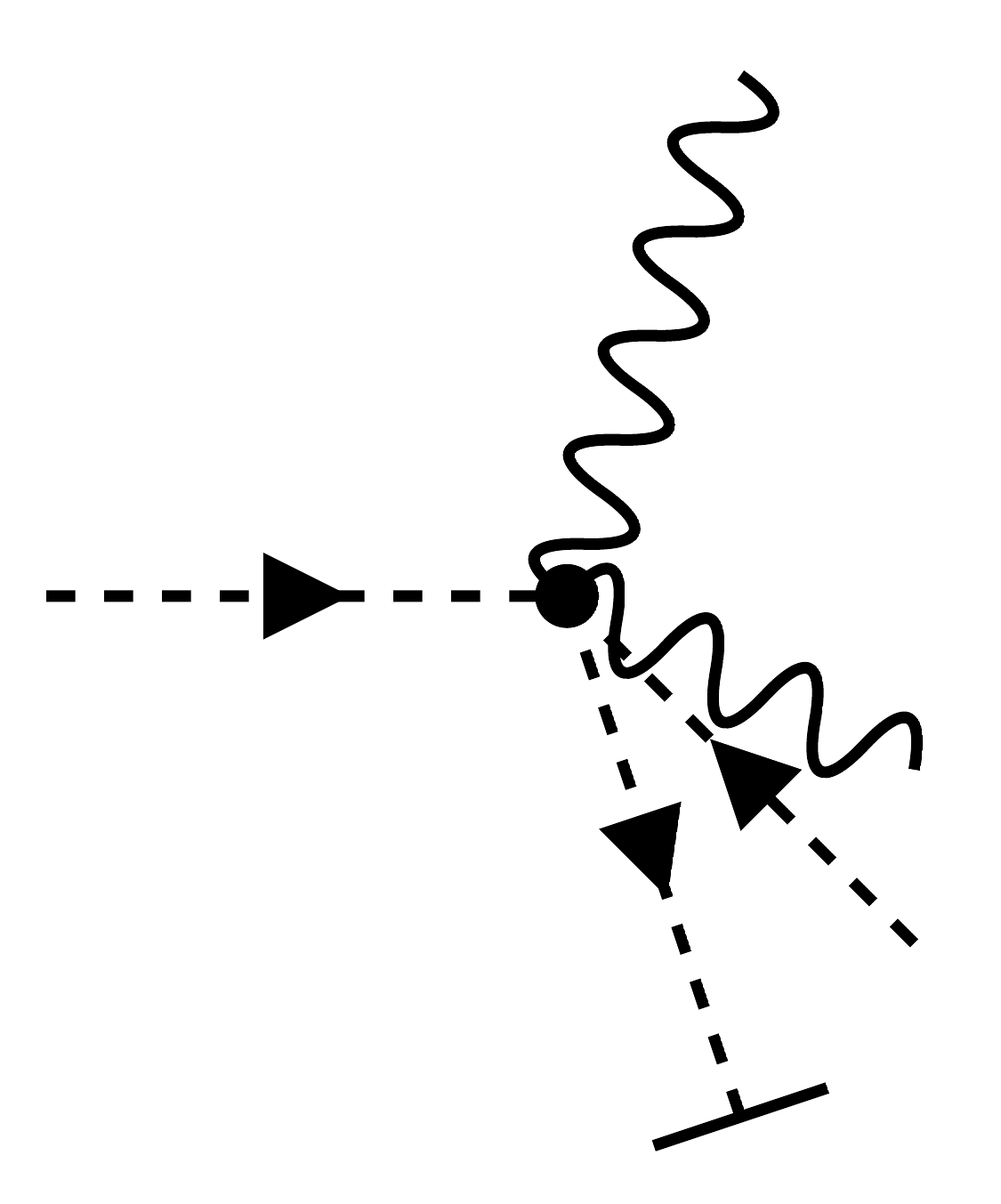} \mkern-6mu \right ) \\
			& \phantom{=} \mkern-50mu + \mkern-5mu \sum_{m+n = i} \left ( \alpha_{n+1} \, \beta_{m+2} \mkern-15mu \graph{4cm}{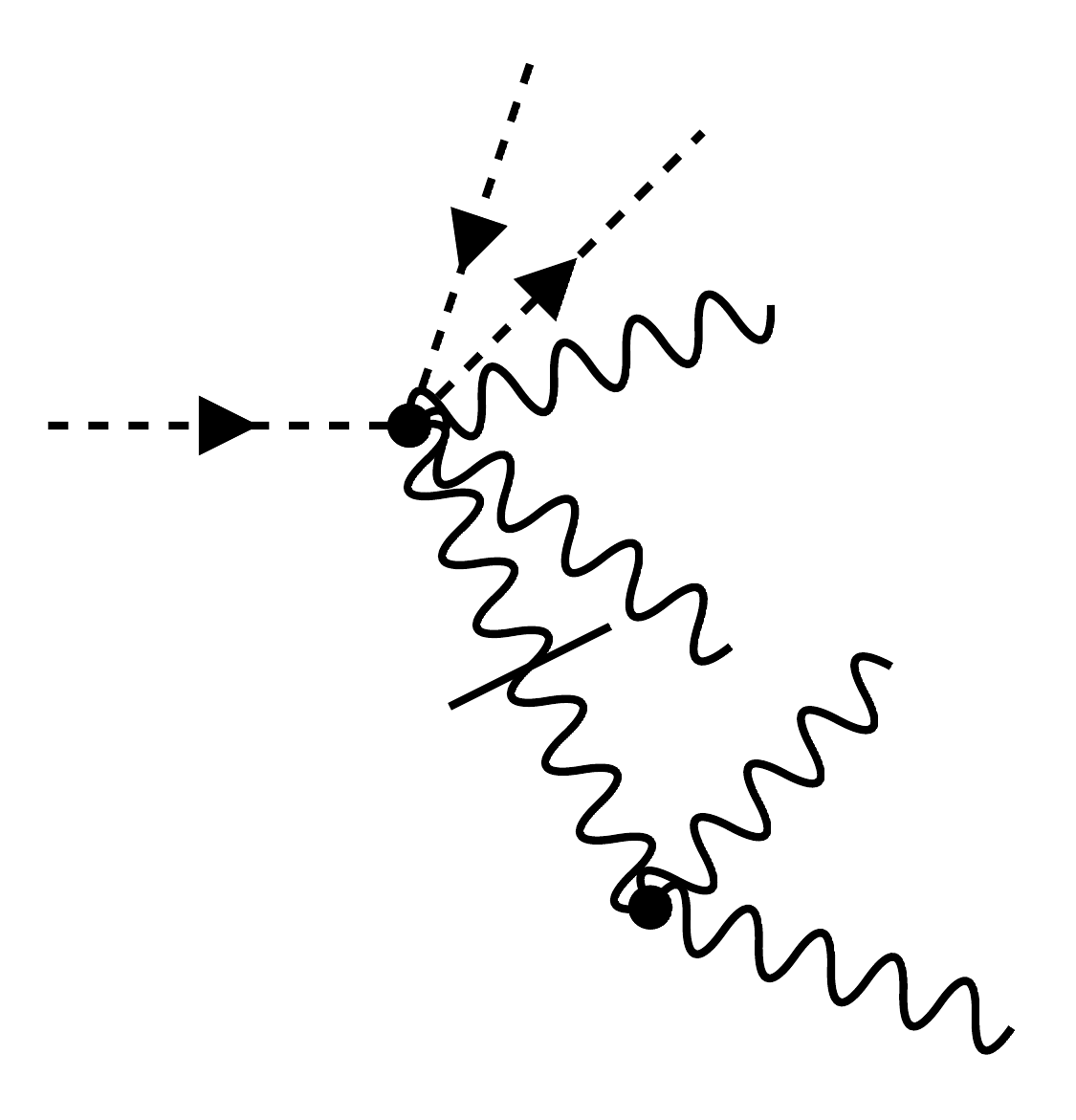}^{\mkern-95mu \raisebox{-17.5mm}{\(\scriptstyle m \Bigr )\)}}_{\mkern-45mu \raisebox{10mm}{\(\scriptstyle n \biggr )\)}} \mkern-30mu + \alpha_{m+2} \, \beta_{n+1} \mkern-15mu \graph{4cm}{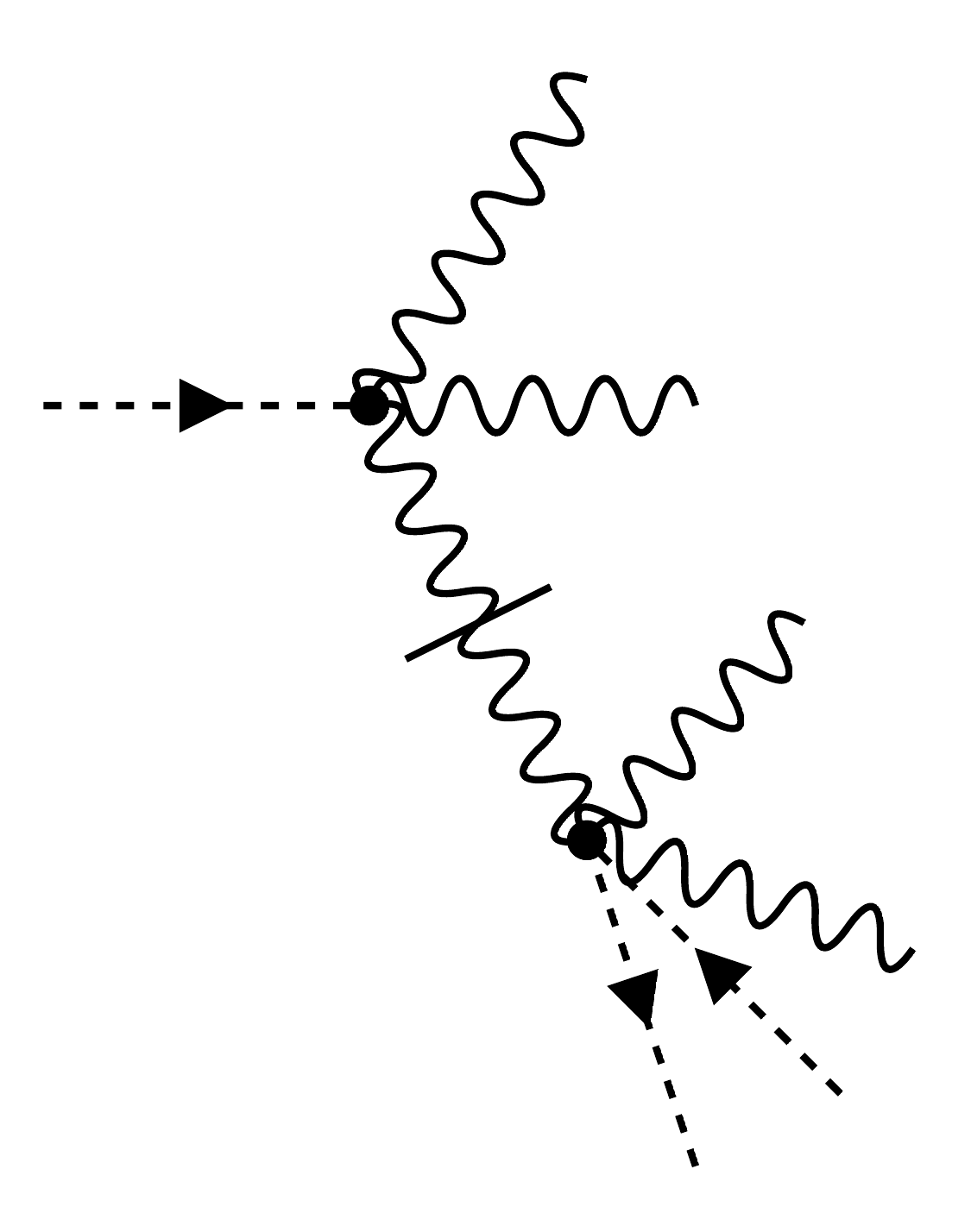}^{\mkern-100mu \raisebox{-10.5mm}{\(\scriptstyle m \Bigr )\)}}_{\mkern-45mu \raisebox{17.5mm}{\(\scriptstyle n \Bigr )\)}} \right . \\
			& \phantom{= \sum (} \left . + \beta_{m+1} \, \beta_n \mkern-15mu \graph{4cm}{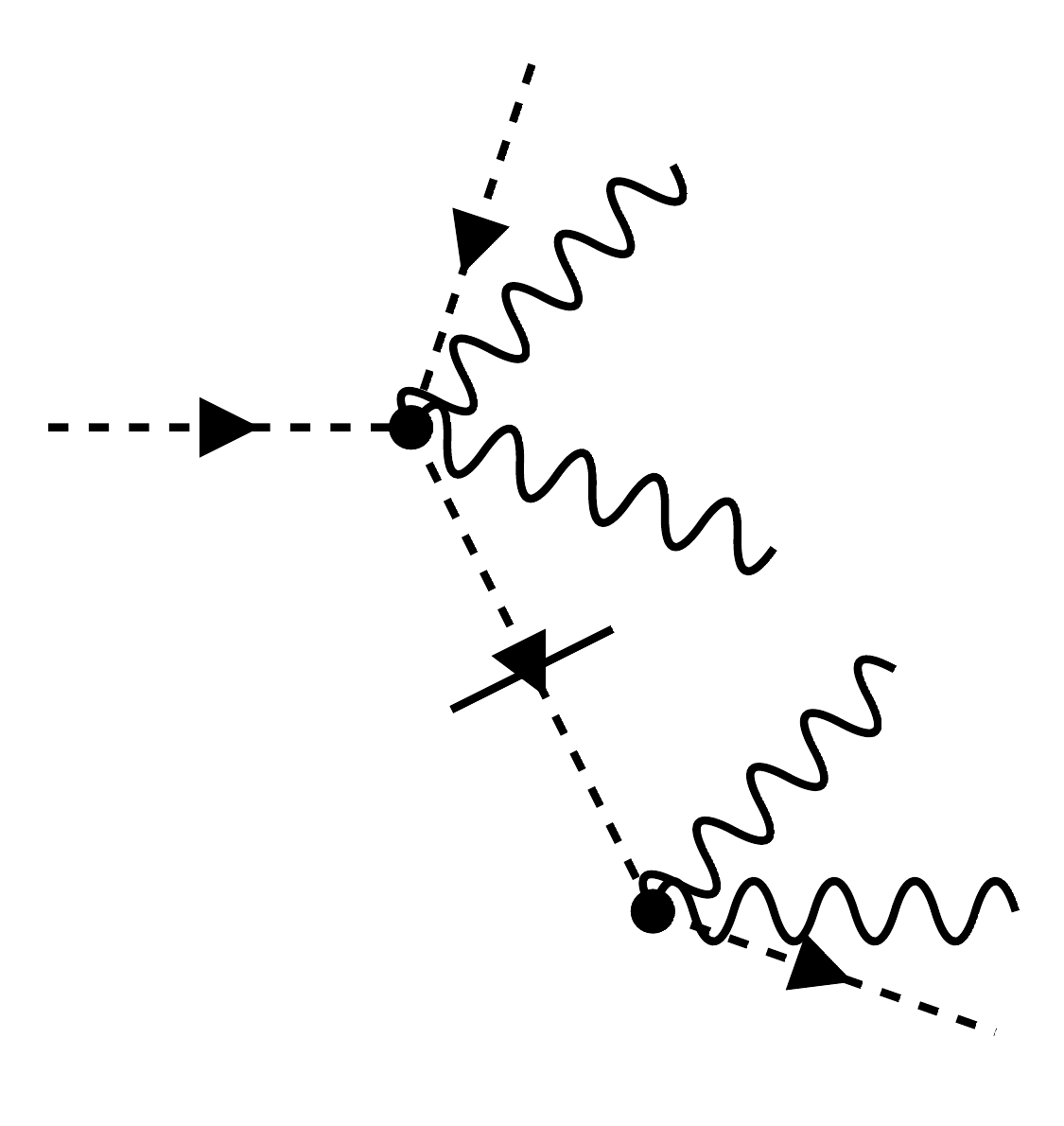}^{\mkern-100mu \raisebox{-13mm}{\(\scriptstyle m \Bigr )\)}}_{\mkern-40mu \raisebox{12.5mm}{\(\scriptstyle n \Bigr )\)}} \mkern-10mu + \beta_{m+1} \, \beta_n \mkern-15mu \graph{4cm}{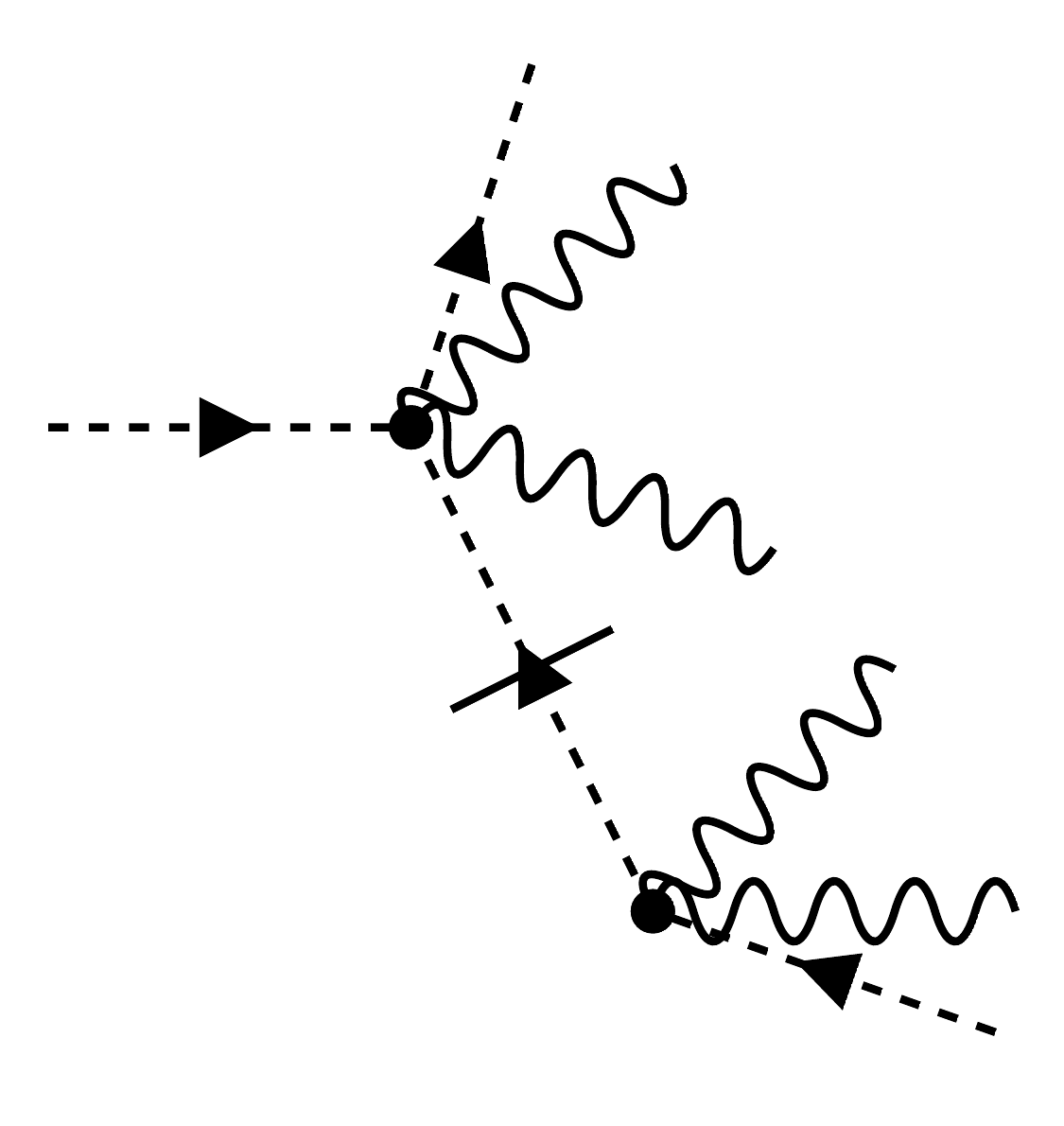}^{\mkern-100mu \raisebox{-13mm}{\(\scriptstyle m \Bigr )\)}}_{\mkern-40mu \raisebox{12.5mm}{\(\scriptstyle n \Bigr )\)}} \right )
		\end{split}
		\intertext{i.e.\ the longitudinal projection of a gauge-ghost vertex is the gauge-ghost vertex with an external on-shell cancellation and all internal on-shell expansions of the vertex,}
		\begin{split}
			\gtproj \mkern-9.84mu \graph{3cm}{gauge-matter-vertex}^{\mkern-50mu \raisebox{-6mm}{\(\scriptstyle i \Bigr )\)}}_{\mkern-50mu \raisebox{6mm}{\(\scriptstyle k \Bigr )\)}} & \simeq_\textup{CI} \sum_{l = 2}^i \left ( \gamma_{i+1,k} \! \oset{\mkern33mu l \mkern-33mu}{\graph{3cm}{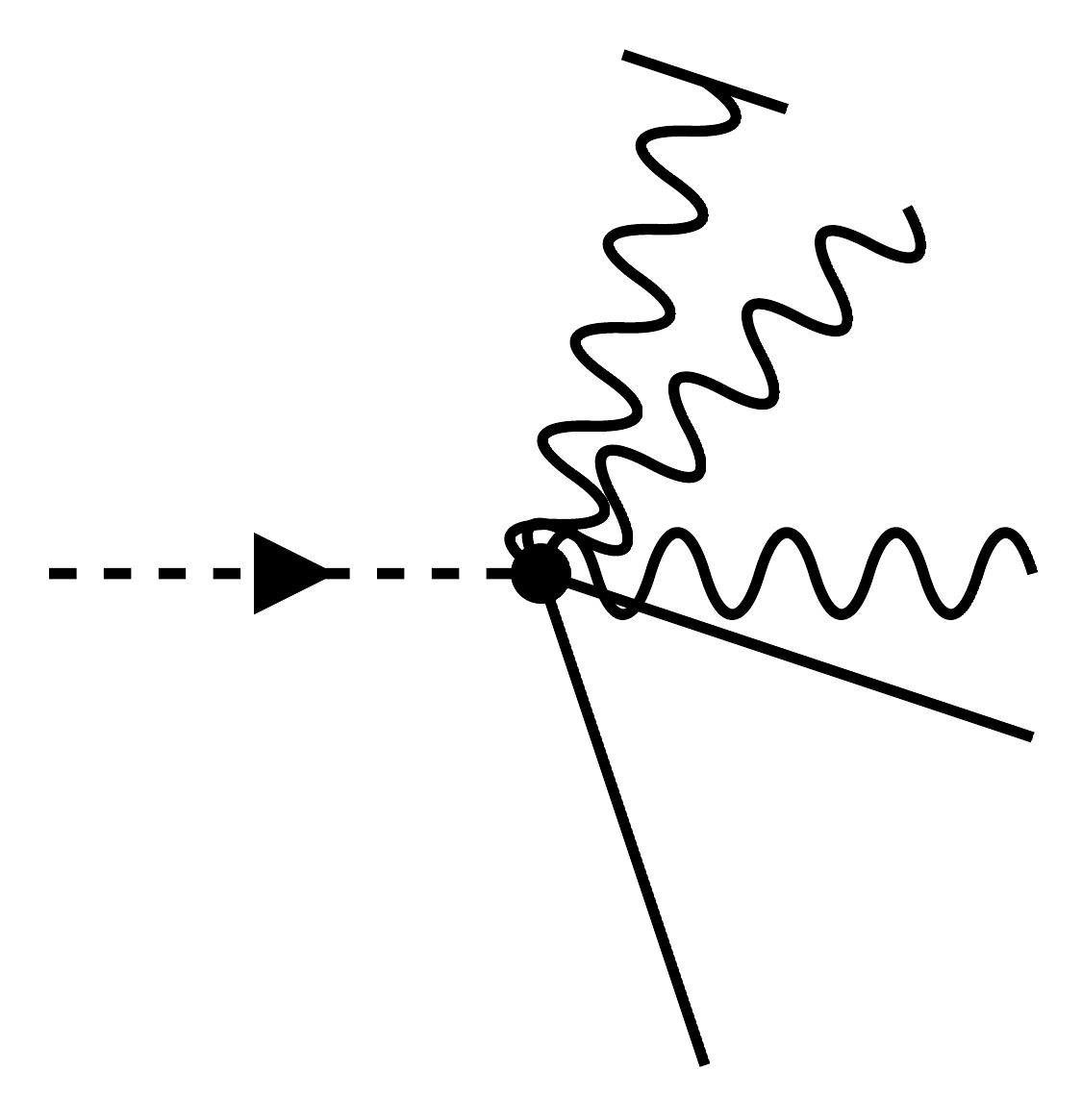}}^{\mkern-40mu \raisebox{-12mm}{\(\scriptstyle i-1 \Bigr )\)}}_{\mkern-50mu \raisebox{6mm}{\(\scriptstyle k \Bigr )\)}} \right ) + \sum_{l = i+1}^{i+k} \left ( \gamma_{i+1,k} \! \uset{\mkern33mu l \mkern-33mu}{\graph{3cm}{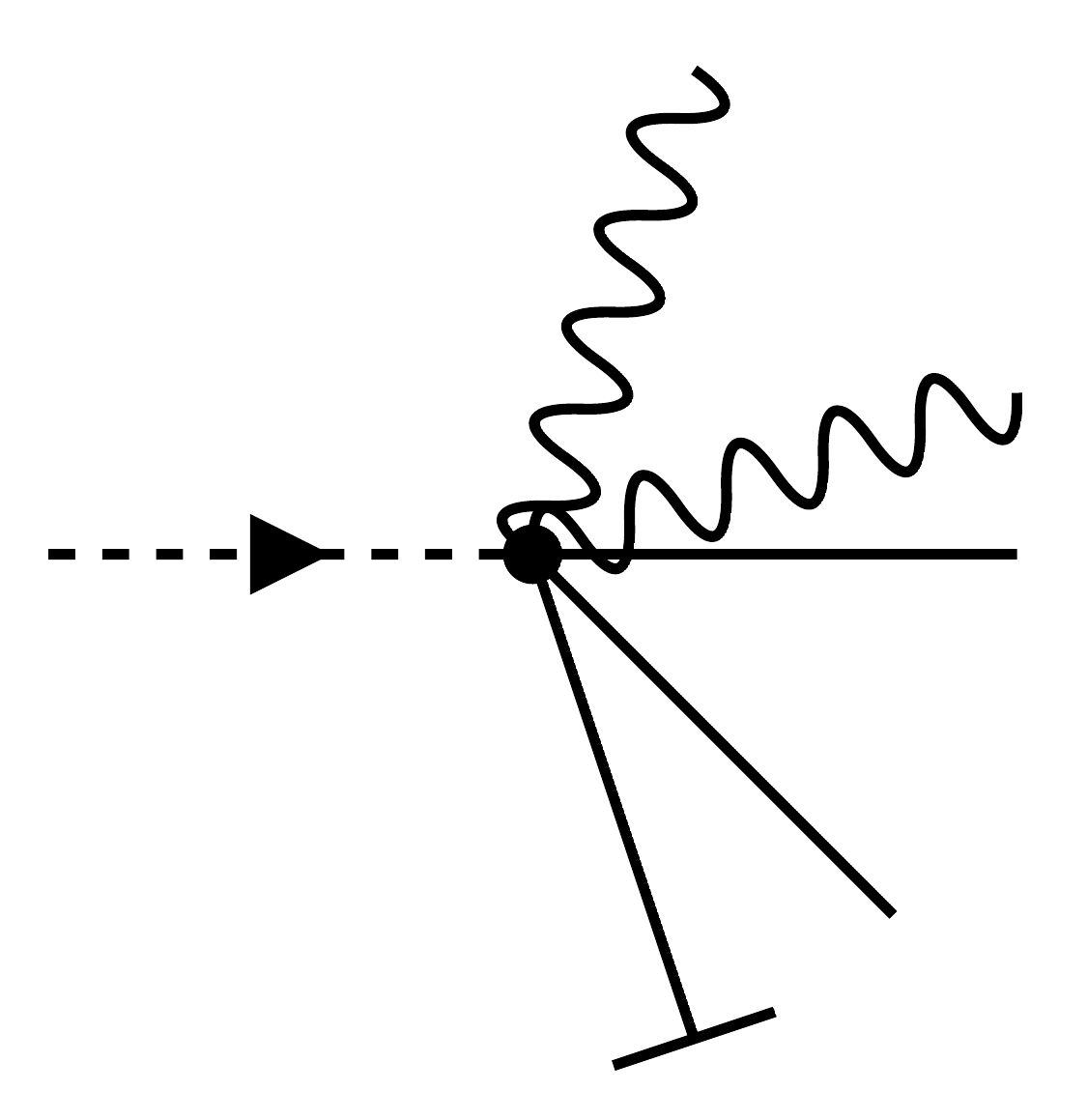}}^{\mkern-50mu \raisebox{-6mm}{\(\scriptstyle i \Bigr )\)}}_{\mkern-40mu \raisebox{12mm}{\(\scriptstyle k-1 \Bigr )\)}} \right ) \\
			& \phantom{=} \mkern-100mu + \mkern-5mu \sum_{\subalign{m+n & = i \\ p+q & = k}} \left ( \gamma_{m+2,p} \, \gamma_{n+1,q} \mkern-15mu \graph{4cm}{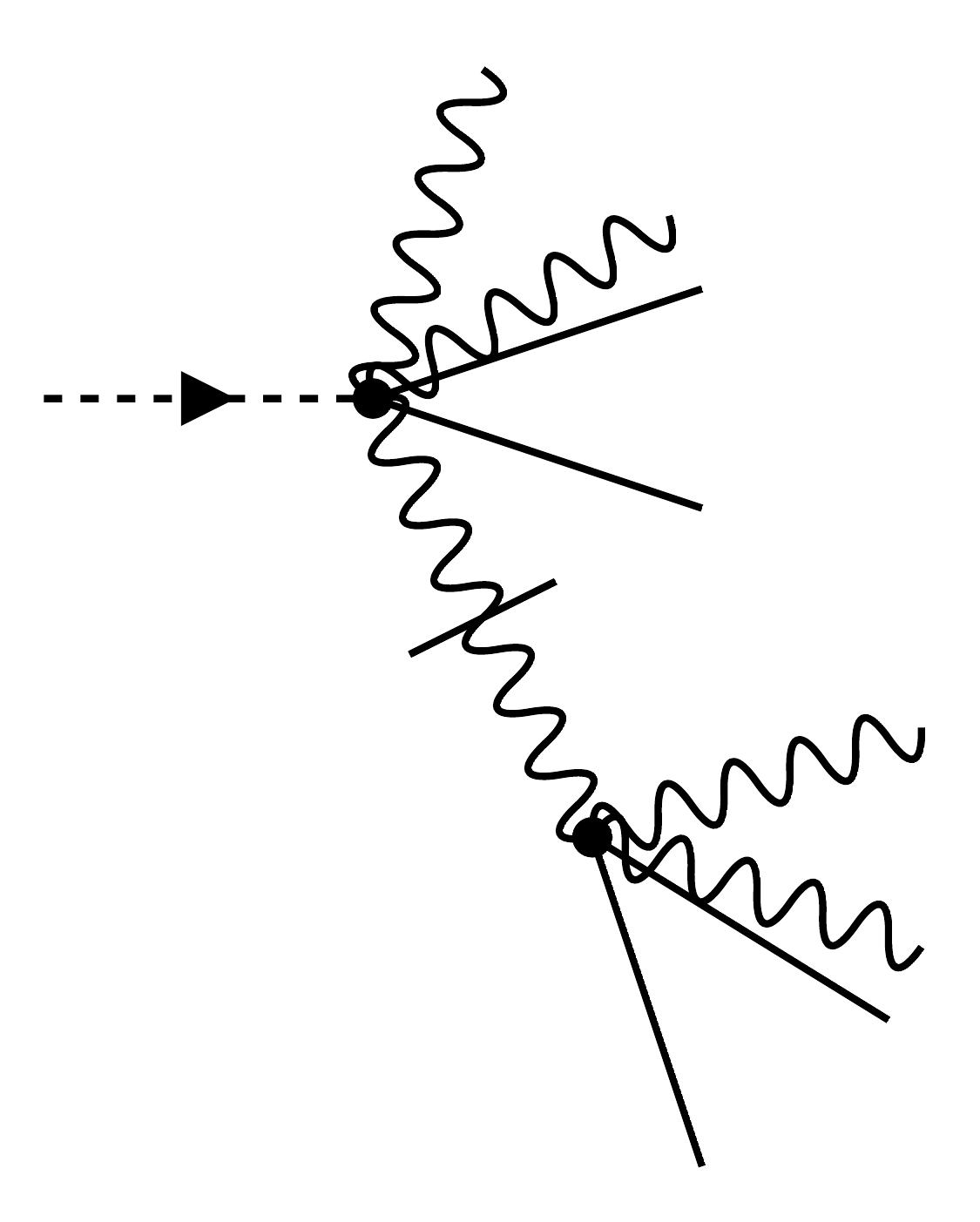}^{\mkern-110mu \raisebox{-5mm}{\(\scriptstyle m \bigr )\)} \mkern-10mu \raisebox{-16mm}{\(\scriptstyle p \bigr )\)}}_{\mkern-60mu \raisebox{6mm}{\(\scriptstyle q \bigr )\)} \mkern10mu \raisebox{15mm}{\(\scriptstyle n \bigr )\)}} \mkern-18mu + \gamma_{m+1,p+1} \, \gamma_{n,q+1} \mkern-15mu \graph{4cm}{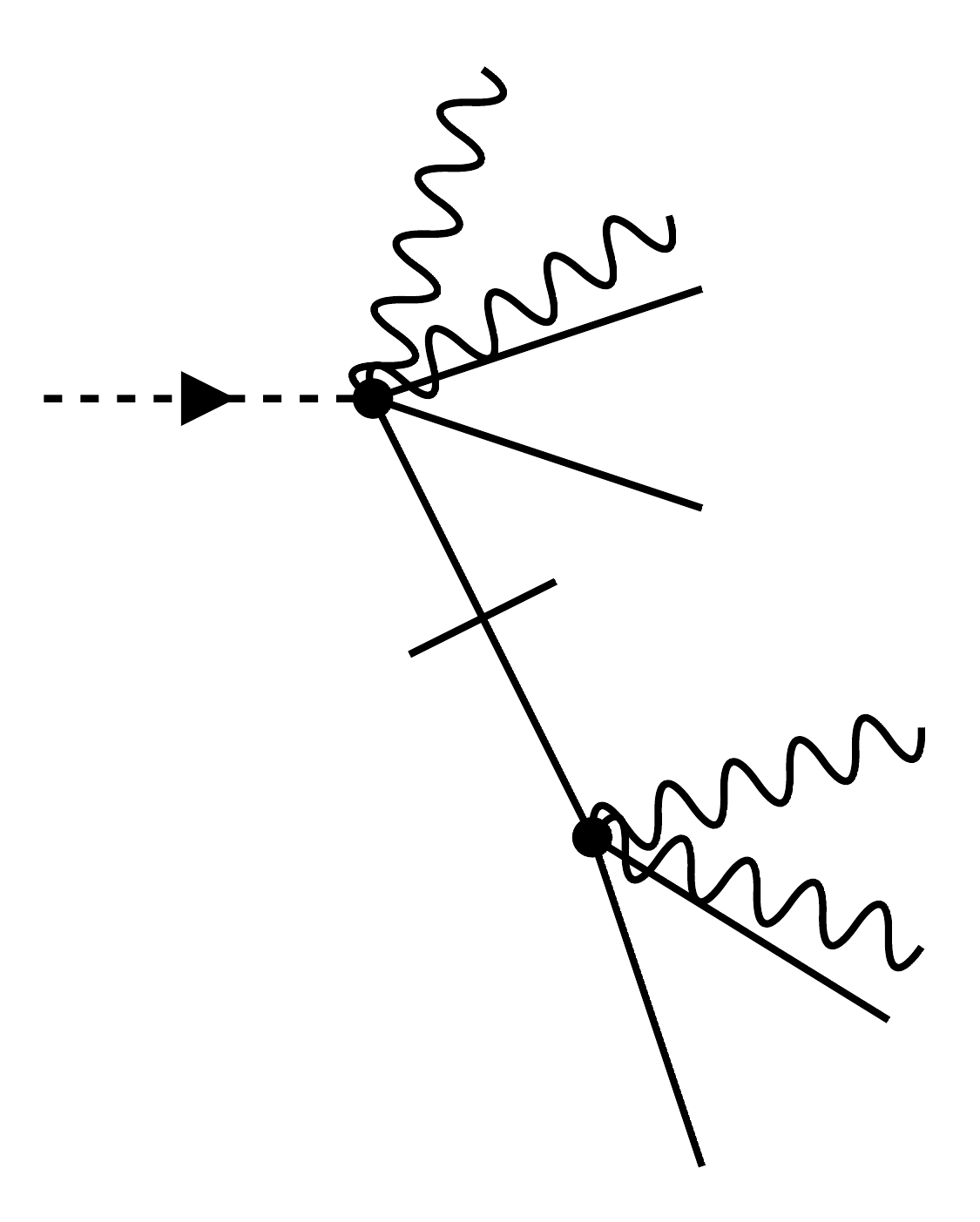}^{\mkern-110mu \raisebox{-5mm}{\(\scriptstyle m \bigr )\)} \mkern-10mu \raisebox{-16mm}{\(\scriptstyle p \bigr )\)}}_{\mkern-60mu \raisebox{6mm}{\(\scriptstyle q \bigr )\)} \mkern10mu \raisebox{15mm}{\(\scriptstyle n \bigr )\)}} \right ) \, ,
		\end{split}
	\end{align}
	\end{subequations}
	}%
	i.e.\ the longitudinal projection of a gauge-matter vertex is the gauge-matter vertex with an external on-shell cancellation and all internal on-shell expansions of the vertex.
\end{thm}

\enter

\begin{proof}
	This statement follows from a combinatorial diagram chase: The contraction of a gauge boson vertices produces either a longitudinal mode traveling through the vertex if there exists a corresponding ghost vertex, it can produce an external on-shell cancellation relating it to higher-valent gauge vertices if existent or it can expand the vertex relating it to lower-valent vertices if existent. The contraction of ghost or matter vertices can also either result in on-shell cancellations of external edges if there exists a higher-valent vertex or in expansions into respective two vertices if there exist the corresponding lower-valent vertices.
\end{proof}

\enter

\begin{exmp}[Quantum Yang--Mills theory] \label{exmp:qym}
	We consider pure Quantum Yang--Mills theory with a Lorenz gauge fixing and a Faddeev--Popov ghost Lagrange density:
	\begin{equation}
	\begin{split}
		\mathcal{L}_\text{QYM} & \coloneq \mathcal{L}_\text{YM} + \mathcal{L}_\text{GF} + \mathcal{L}_\text{Ghost} \\ & \phantom{:} \equiv - \eta^{\mu \nu} \eta^{\rho \sigma} \delta_{a b} \left ( \frac{1}{4 \mathrm{g}^2} F^a_{\mu \rho} F^b_{\nu \sigma} + \frac{1}{2 \xi} \big ( \partial_\mu A^a_\nu \big ) \big ( \partial_\rho A^b_\sigma \big ) \right ) \dif V_\eta \\
		& \phantom{\coloneq} + \eta^{\mu \nu} \left ( \frac{1}{\xi} \overline{c}_a \left ( \partial_\mu \partial_\nu c^a \right ) + \mathrm{g} \tensor{f}{^a _b _c} \overline{c}_a \left ( \partial_\mu \big ( c^b A^c_\nu \big ) \right ) \right ) \dif V_\eta \, ,
	\end{split}
	\end{equation}
	where \(F^a_{\mu \nu} := \mathrm{g} \big ( \partial_\mu A^a_\nu - \partial_\nu A^a_\mu \big ) - \mathrm{g}^2 \tensor{f}{^a _b _c} A^b_\mu A^c_\nu\) is the local curvature form of the gauge boson \(A^a_\mu\). Furthermore, \(\dif V_\eta := \dif t \wedge \dif x \wedge \dif y \wedge \dif z\) denotes the Minkowskian volume form. Additionally, \(\eta^{\mu \nu} \partial_\mu A^a_\nu \equiv 0\) is the Lorenz gauge fixing functional and \(\xi\) the gauge fixing parameter. Finally, \(c^a\) and \(\overline{c}_a\) are the gauge ghost and gauge antighost, respectively. Specifically, for the Lorenz gauge fixing, we have the following transversal structure:
	\begin{subequations} \label{eqn:projection_tensors_qym}
	\begin{align}
		L^\nu_\mu & := \frac{1}{p^2} p^\nu p_\mu \, , \\
		I^\nu_\mu & := \delta^\nu_\mu
		\intertext{and}
		T^\nu_\mu & := I^\nu_\mu - L^\nu_\mu
	\end{align}
	\end{subequations}
	as longitudinal, identical and transversal projection tensors, where a short calculation verifies \(L^2 = L\), \(I^2 = I\) and \(T^2 = T\). Specifically, the longitudinal and gauge projection (co)vectors are then given as eigenvectors to the transversal structure, cf.\ \cite[Section 3.1]{Prinz_7}:
	\begin{subequations}
	\begin{align}
		\gtproj_\mu &\equiv g_\mu \coloneq \frac{1}{p^2} p_\mu
		\intertext{and}
		\gfproj^\nu & \equiv l^\nu \coloneq p^\nu
	\end{align}
	\end{subequations}
	This corresponds to the situation in \thmref{thm:gen-ci} with the following constants being non-zero: \(\alpha_3 = \alpha_4 = 1\) because there are only three- and four-valent gauge boson vertices and \(\beta_1 = 1\) because there is only a three-valent ghost-vertex.
\end{exmp}

\enter

\begin{exmp}[(Effective) Quantum General Relativity] \label{exmp:qgr}
	We consider pure (effective) Quantum General Relativity with a de Donder gauge fixing and a Faddeev--Popov ghost Lagrange density:
	\begin{equation}
	\begin{split}
		\mathcal{L}_\text{QGR} & \coloneq \mathcal{L}_\text{GR} + \mathcal{L}_\text{GF} + \mathcal{L}_\text{Ghost} \\ & \phantom{:} \equiv - \frac{1}{2 \varkappa^2} \left ( \sqrt{- \Det{g}} R + \frac{1}{2 \zeta} \eta^{\mu \nu} \deDonder^{(1)}_\mu \deDonder^{(1)}_\nu \right ) \dif V_\eta \\ & \phantom{\coloneq} - \frac{1}{2} \eta^{\rho \sigma} \left ( \frac{1}{\zeta} \overline{C}^\mu \left ( \partial_\rho \partial_\sigma C_\mu \right ) + \overline{C}^\mu \left ( \partial_\mu \big ( \tensor{\Gamma}{^\nu _\rho _\sigma} C_\nu \big ) - 2 \partial_\rho \big ( \tensor{\Gamma}{^\nu _\mu _\sigma} C_\nu \big ) \right ) \right ) \dif V_\eta \, ,
	\end{split}
	\end{equation}
	where \(R := g^{\nu \sigma} \tensor{R}{^\mu _\nu _\mu _\sigma}\) is the Ricci scalar (with \(\tensor{R}{^\rho _\sigma _\mu _\nu} := \partial_\mu \tensor{\Gamma}{^\rho _\nu _\sigma} - \partial_\nu \tensor{\Gamma}{^\rho _\mu _\sigma} + \tensor{\Gamma}{^\rho _\mu _\lambda} \tensor{\Gamma}{^\lambda _\nu _\sigma} - \tensor{\Gamma}{^\rho _\nu _\lambda} \tensor{\Gamma}{^\lambda _\mu _\sigma}\) the Riemann tensor). Again, \(\dif V_\eta := \dif t \wedge \dif x \wedge \dif y \wedge \dif z\) denotes the Minkowskian volume form, which is related to the Riemannian volume form \(\dif V_g\) via \(\dif V_g \equiv \sqrt{- \Det{g}} \dif V_\eta\). Additionally, \(\deDonder^{(1)}_\mu := \eta^{\rho \sigma} \Gamma_{\mu \rho \sigma} \equiv 0\) is the linearized de Donder gauge fixing functional and \(\zeta\) the gauge fixing parameter. Finally, \(C_\mu\) and \(\overline{C}^\mu\) are the graviton-ghost and graviton-antighost, respectively. Again, we refer to \cite{Prinz_4,Prinz_2,Prinz_8} for more detailed introductions and further comments on the chosen conventions. Specifically, for the de Donder gauge fixing, we have the following transversal structure:
	\begin{subequations} \label{eqn:projection_tensors_qgr}
	\begin{align}
		\bbL^{\rho \sigma}_{\mu \nu} & := \frac{1}{2 p^2} \left ( \delta^\rho_\mu p^\sigma p_\nu + \delta^\sigma_\mu p^\rho p_\nu + \delta^\rho_\nu p^\sigma p_\mu + \delta^\sigma_\nu p^\rho p_\mu - 2 \eta^{\rho \sigma} p_\mu p_\nu \right ) \, , \\
		\bbI^{\rho \sigma}_{\mu \nu} & := \frac{1}{2} \left ( \delta^\rho_\mu \delta^\sigma_\nu + \delta^\sigma_\mu \delta^\rho_\nu \right )
		\intertext{and}
		\bbT^{\rho \sigma}_{\mu \nu} & := \bbI^{\rho \sigma}_{\mu \nu} - \bbL^{\rho \sigma}_{\mu \nu} \, ,
	\end{align}
	\end{subequations}
	as longitudinal, identical and transversal projection tensors, where a short calculation verifies \(\bbL^2 = \bbL\), \(\bbI^2 = \bbI\) and \(\bbT^2 = \bbT\). Specifically, the longitudinal and gauge projection tensors are then given as eigentensors to the transversal structure, cf.\ \cite[Section 3.2]{Prinz_7}:
	\begin{subequations}
	\begin{align}
		\gtproj_{\mu \nu}^\kappa &\equiv \mathscr{G}_{\mu \nu}^\kappa \coloneq \frac{1}{p^2} \big ( p_\mu \delta_\nu^\kappa + p_\nu \delta_\mu^\kappa \big )
		\intertext{and}
		\gfproj^{\rho \sigma}_\lambda & \equiv \mathscr{L}^{\rho \sigma}_\lambda \coloneq \frac{1}{2} \big ( p^\rho \delta^\sigma_\lambda + p^\sigma \delta^\rho_\lambda - p_\lambda \eta^{\rho \sigma} \big )
	\end{align}
	\end{subequations}
	This corresponds to the situation in \thmref{thm:gen-ci} with \(\alpha_i = 1\) for all \(i \in \mathbb{N}_{\geq 3}\) and \(\beta_i = 1\) for all \(i \in \mathbb{N}_{\geq 1}\). Interestingly, for the metric density decomposition \(\sqrt{- \dt{g}} g^{\mu \nu} \equiv \eta^{\mu \nu} + \varkappa \boldsymbol{\phi}^{\mu \nu}\) of Goldberg \cite{Goldberg} and Capper et al.\ \cite{Capper_Leibbrandt_Ramon-Medrano,Capper_Medrano,Capper_Namazie} there exists only a three-valent ghost-vertex and thus \(\beta_1 = 1\) with \(\beta_i = 0\) for all \(i \in \mathbb{N}_{\geq 2}\), cf.\ \cite{Kissler,Prinz_7}.
\end{exmp}

\enter

\begin{defn}[Cancellation identities relation] \label{defn:cancellation-identities-relation}
	Let \(\GQc\) be the set of connected Feynman graphs and \(\Gamma_1, \Gamma_2 \in \GQc\) be two such Feynman graphs with respective subtrees \(\tau_1 \subseteq \Gamma_1\) and \(\tau_2 \subseteq \Gamma_2\). We introduce the following equivalence relation:
	\begin{equation}
		\Gamma_1 \sim_\text{CI} \Gamma_2 \quad \iff \quad \tau_1 \simeq_\text{CI} \tau_2 \, ,
	\end{equation}
	where \(\simeq_\text{CI}\) denotes equality of two trees with respect to the identities from \eqnsref{eqns:cancellation-identitites}.
\end{defn}

\section{The pBRST Feynman graph complex} \label{sec:pBRST-Feynman-graph-complex}

Now, we consider the set of connected Feynman diagrams \(\GQc\) and construct a Feynman graph complex, which serves as a perturbative BRST (pBRST) complex. More specifically, like the BRST differential, which replaces fields via their gauge transformation in direction of a ghost field, we project external gauge boson edges onto their longitudinal degrees of freedom. Specifically, the cohomology groups are then by construction transversal (i.e.\ gauge invariant) linear combinations of Feynman graphs. Using the formalized cancellation identities from \sectionref{sec:formalized-cancellation-identities}, we find that the even cohomology groups are given via connected Green's functions with corresponding number of external ghosts, while the odd cohomology groups vanish.

\enter

\begin{defn}[Connected Feynman graphs with ordered and labeled external edges] \label{defn:connected-colored-ordered-feynman-graph}
	Let \(\GQc\) be the set of connected Feynman graphs from \defnref{defn:combinatorial-data-qft} and \(\Gamma \in \GQc\) a Feynman graph therein. We denote by \(\EG{\Gamma}\) the set of its external gauge boson edges, which we order by means of the map \(\sigma \colon \EG{\Gamma} \to \mathbb{N}\). Moreover, we label the set \(\EG{\Gamma}\) by elements of the transversal structure \(\TQ \coloneq \set{\boldsymbol{L}, \boldsymbol{I}, \boldsymbol{T}}_{\boldsymbol{G}}\) via the map \(\tau \colon \EG{\Gamma} \to \TQ\); the default is that external edges are labeled by \(\boldsymbol{T}\). From this, we construct the set of connected Feynman graphs with ordered and labeled external edges as triples \(\GGamma \equiv (\Gamma, \sigma, \tau)\) in the set \(\GGamma \in \GQo\).
\end{defn}

\enter

\begin{defn}[Gaugeon-grading] \label{defn:gaugeon-grading}
	Let \(\GGamma \in \GQo\) as in \defnref{defn:connected-colored-ordered-feynman-graph}. We introduce the map
	\begin{equation}
		\mathfrak{L} \, : \quad \GQo \to \mathbb{N}_0 \, , \quad \GGamma \mapsto \left ( \# \text{ \(\boldsymbol{L}\)-colored external edges} \right ) + \left ( \# \text{ external ghost edges} \right ) \, ,
	\end{equation}
	which we call \emph{gaugeon-grading}.
\end{defn}

\enter

\begin{defn}[Gaugeon graph complex] \label{defn:gaugeon-graph-complex}
	Given the situation of \defnsaref{defn:connected-colored-ordered-feynman-graph}{defn:gaugeon-grading}. Then we introduce the longitudinal projection operator for the external gauge boson edge \(e \in \EG{\Gamma}\) via (no Einstein summation over \(e\))
	\begin{subequations}
	\begin{align}
		\lambda_e \, & : \quad \GQo \to \GQo \, , \quad \GGamma \mapsto \left ( \boldsymbol{L}_e \frac{\partial}{\partial \boldsymbol{T}_e} \right ) \GGamma \, .
		\intertext{Additionally, we introduce the \emph{perturbative BRST differential}}
		\llambda \GGamma \, & : \quad \GQo \to \GQo \, , \quad \GGamma \mapsto \sum_{e \in \EG{\GGamma}} \left ( -1 \right )^{\mathfrak{L} \left ( \GGamma ; e \right )} \lambda_e \GGamma \, , \label{eqn:gaugeon-projection-operator}
	\end{align}
	\end{subequations}
	where \(\mathfrak{L} \left ( \GGamma ; e \right ) \colon \EG{\GGamma} \to \mathbb{N}_0\) denotes the gaugeon-grading of \(\GGamma\) until external edge \(e \in \EG{\Gamma}\) according to the chosen ordering \(\sigma \left ( \GGamma \right )\). Finally, we denote by \(\CQ \coloneq (\GQo, \llambda)\) the corresponding Feynman graph complex.
\end{defn}

\enter

\begin{lem} \label{lem:lambda-cohomological-differential}
	The operator \(\llambda\) is a cohomological differential for the gaugeon-grading with trivial cohomology.
\end{lem}

\begin{proof}
	For the first part, we observe that \(\llambda\) squares to zero due to the sign introduced by \(\mathfrak{L} \left ( \GGamma ; e \right )\) in \eqnref{eqn:gaugeon-projection-operator} and the observation that \(\mathfrak{L} \left ( \llambda \GGamma \right ) = \mathfrak{L} \left ( \GGamma \right ) + 1\), with \(\GGamma \in \GQo\). Then, for the second part, let \(\ggamma\) be a cycle, i.e.\ \(\llambda \ggamma = 0\). We will now show that \(\ggamma\) is then already a boundary, i.e.\ there exists an element \(\boldsymbol{\gamma^\prime}\) such that \(\ggamma = \llambda \boldsymbol{\gamma^\prime}\). For this, we assume without loss of generality that \(\ggamma\) consists of linear combinations of graphs with \(m\) external edges and that it is homogeneous with respect to the gaugeon-grading, i.e.\ \(\mathfrak{L} \left ( \ggamma \right ) = n\), as only homogeneous graphs can cancel themselves anyway. In order to satisfy \(\llambda \ggamma = 0\), we need that \(\ggamma\) is the sum of \(m\)-divisible many graphs, such that every external gauge boson edge possesses one summand with an  \(\boldsymbol{L}\)-label attached to it. Additionally, depending on whether \(n\) is even or odd, we also need according signs for \(\ggamma\) to be in the kernel of \(\llambda\). But these are precisely the conditions for \(\ggamma\) being in the image of \(\llambda\) and thus all cohomology groups of the corresponding complex vanish.
\end{proof}

\enter

\begin{rem}
	After having introduced the general graph complex with trivial cohomology, we now add physics by implementing the cancellation identities of \sectionref{sec:formalized-cancellation-identities}. This will also introduce interesting cohomology on the graph complex, related to transversal linear combinations of Feynman graphs as we will discuss below.
\end{rem}

\enter

\begin{defn}[pBRST Feynman graph complex]
	Let \(\mathcal{C}_\Q\) be the \emph{gaugeon graph complex} from \defnref{defn:gaugeon-graph-complex} and \(\sim_\text{CI}\) be the \emph{cancelation identities relation} from \defnref{defn:cancellation-identities-relation}. Now, we define the \emph{pBRST graph complex} as the gaugeon graph complex modulo the cancellation identities relation, i.e.\ \(\PQ \coloneq \CQ / \sim_\text{CI}\).
\end{defn}

\enter

\begin{thm}[pBRST cohomology] \label{thm:pBRST-cohomology}
	The cohomology groups of the pBRST Feynman graph complex \(\, \PQ\) are generated as follows:
	\begin{equation}
		H^j \left ( \PQ, \mathbb{Z} \right ) \equiv \begin{cases} \left \langle \rescombgreenc^{i, j, k}_\mathbf{c} \right \rangle_{i, k, \mathbf{c}} & \text{if \(j\) is even} \\
		0 & \text{else} \end{cases}
	\end{equation}
	In particular, the generators are connected restricted combinatorial Green's functions with arbitrary number of external gauge bosons, matter particles and coupling-grading.
\end{thm}

\begin{proof}
	This follows directly from \thmref{thm:gen-ci}, \lemref{lem:lambda-cohomological-differential} and the general construction.
\end{proof}

\section{Derived renormalization theory} \label{sec:derived-renormalization-theory}

Following on the result of \thmref{thm:pBRST-cohomology}, we now construct the renormalization Hopf algebra on its cohomology groups, rather than on 1PI Feynman graphs. This amounts to a \emph{derived version} of \emph{Connes--Kreimer renormalization theory}. Furthermore, it also connects directly to the well-known \emph{Hopf subalgebras}, generated by combinatorial Green's functions, and thus to \emph{multiplicative renormalization} as we will this discuss in this section.

\enter

\begin{defn}[Derived renormalization Hopf algebra] \label{defn:derived-renormalization-theory}
	Let \(\Q\) be a (generalized) Quantum Gauge Theory with labeled, oriented and connected Feynman graph set \(\GQo\) and pBRST Feynman graph complex \(\CQ\). We denote by \(\HQd \equiv (\VQd, m, \one, \Delta, \coone, S)\) the derived renormalization Hopf algebra: Contrary to the ordinary construction in \defnref{defn:renormalization_hopf_algebra}, the algebra part is now generated via power series in the cohomology groups of \(\PQ\) instead of 1PI Feynman graphs, i.e.\ \(\VQd \coloneq \mathbb{Q} \big [ \mkern-1mu \bigl [ \rescombgreenc^{i, j, k}_\mathbf{c} \bigr ] \mkern-1mu \big ]_{i, j, k, \mathbf{c}}\), where we have used \thmref{thm:pBRST-cohomology}.
\end{defn}

\enter

\begin{col}[Derived renormalization theory and multiplicative renormalization] \label{col:derived-renormalization-theory}
	If \(\Q\) does not possess gauge anomalies, the construction of \defnref{defn:derived-renormalization-theory} is well-defined and consists of the Hopf subalgebras for multiplicative renormalization, cf.\ \defnref{defn:hopf_subalgebras_renormalization_hopf_algebra}.
\end{col}

\begin{proof}
	A gauge anomaly would correspond to a violation of the cancellation identities in the sense of \thmref{thm:gen-ci}: Thus, if \(\Q\) does not possess gauge anomalies, then \thmref{thm:pBRST-cohomology} holds and its cohomology groups are generated by connected restricted combinatorial Green's functions \(\rescombgreenc^{i, j, k}_\mathbf{c}\) with arbitrary number of external gauge bosons, matter particles and coupling-grading. Furthermore, a second consequence of \thmref{thm:gen-ci} is that the coproduct factors over said combinatorial Green's functions, because the extracted polynomials \(\widehat{\mathfrak{P}}_{\mathbf{g}} \left ( \rescombgreenc^r_{\mathbf{G}} \right )\) are gauge invariant by construction. This is normally not guaranteed for coupling-grading, cf.\ \cite[Corollary 5.8]{Prinz_3}.
\end{proof}

\enter

\begin{rem}
	Unfortunately, it seems not possible to construct directly a differential-graded renormalization Hopf algebra, i.e.\ start from the renormalization Hopf algebra \(\HQ\) and equip it with the differential \(\llambda\): This is due to the fact that the cancellation identities relation of \defnref{defn:cancellation-identities-relation} is not compatible with the extraction of subgraphs. However, on the level of cohomology groups this construction is fine, as they combine to a transversal linear combination of Feynman graphs by construction. Thus, \defnref{defn:derived-renormalization-theory} describes the correct setup to obtain a manifestly gauge invariant renormalization operation.
\end{rem}

\section{Conclusion} \label{sec:conclusion}

Starting from the \emph{Hopf algebraic renormalization} of Connes and Kreimer and the \emph{cancellation identities} of 't Hooft and Veltman, we have developed a further perspective on the renormalization of (generalized) gauge theories. To this end, we have reviewed and generalized the Connes--Kreimer renormalization theory for \emph{connected} Feynman graphs in \sectionref{sec:conn-hopf-alg-ren}. Then, in \sectionref{sec:formalized-cancellation-identities}, we have reviewed and formalized the cancellation identities of 't Hooft and Veltman. In particular, we have studied which contributions of the longitudinal contractions of gauge-vertices cancel against ghost-vertices and which contributions cancel against higher-valent vertices. This generality contains in particular Quantum Yang--Mills theory, but also (effective) Quantum General Relativity, which requires this generalized setup, cf.\ \exmpsaref{exmp:qym}{exmp:qgr}. Next, in \sectionref{sec:pBRST-Feynman-graph-complex}, we have implemented these formalized cancellation identities into a Feynman graph complex, which we call \emph{pBRST Feynman graph complex} due to its structural similarities with the BRST complex. Finally, in \sectionref{sec:derived-renormalization-theory}, we have combined Hopf algebraic renormalization theory with the pBRST Feynman graph complex by constructing a \emph{derived renormalization Hopf algebra} on the cohomology groups of the pBRST Feynman graph complex. In future work, we aim to study the relation of this construction to the BV formalism and want to apply it to UV completions of (effective) Quantum General Relativity.

\section*{Acknowledgments}
\addcontentsline{toc}{section}{Acknowledgments}

The author thanks David Carchedi, Jonah Epstein and Owen Gwilliam for illuminating and helpful discussions. Additionally, the author thanks Albrecht Klemm, Dirk Kreimer and Peter Teichner for general input and supportive mentoring. Finally, the author thanks the \emph{SwissMAP Research Station} in Les Diablerets for several enjoyable conference stays, where many of the key insights for the present work unfolded itself in the beautiful surroundings of the Swiss alps. This research is supported by a postdoctoral fellowship of the \emph{Max Planck Institute for Mathematics}.

\bibliography{References}{}

\begin{thebibliography}{10}
  \providebibliographyfont{name}{}%
  \providebibliographyfont{lastname}{}%
  \providebibliographyfont{title}{\emph}%
  \providebibliographyfont{jtitle}{\btxtitlefont}%
  \providebibliographyfont{etal}{\emph}%
  \providebibliographyfont{journal}{}%
  \providebibliographyfont{volume}{}%
  \providebibliographyfont{ISBN}{\MakeUppercase}%
  \providebibliographyfont{ISSN}{\MakeUppercase}%
  \providebibliographyfont{url}{\url}%
  \providebibliographyfont{numeral}{}%
  \expandafter\btxselectlanguage\expandafter {\btxfallbacklanguage}

\btxselectlanguage {english}
\bibitem {Kreimer_Hopf_Algebra}
\btxnamefont {\btxlastnamefont {{D. Kreimer}}}\btxauthorcolon\ \btxjtitlefont
  {\btxifchangecase {{O}n the {H}opf algebra structure of perturbative quantum
  field theories}{{O}n the {H}opf algebra structure of perturbative quantum
  field theories}}.
\newblock \btxjournalfont {Adv. Theor. Math. Phys.}, 2:303--334, 1998.
\newblock arXiv:q-alg/9707029v4.

\bibitem {Connes_Kreimer_NG}
\btxnamefont {\btxlastnamefont {{A. Connes}}} \btxandlong {}\ \btxnamefont
  {\btxlastnamefont {{D. Kreimer}}}\btxauthorcolon\ \btxjtitlefont
  {\btxifchangecase {{H}opf {A}lgebras, {R}enormalization and {N}oncommutative
  {G}eometry}{{H}opf {A}lgebras, {R}enormalization and {N}oncommutative
  {G}eometry}}.
\newblock \btxjournalfont {Commun. Math. Phys.}, 199:203--242, 1998.
\newblock arXiv:hep-th/9808042v1.

\bibitem {Broadhurst_Kreimer}
\btxnamefont {\btxlastnamefont {{D. J. Broadhurst}}} \btxandlong {}\
  \btxnamefont {\btxlastnamefont {{D. Kreimer}}}\btxauthorcolon\ \btxjtitlefont
  {\btxifchangecase {{R}enormalization automated by {H}opf
  algebra}{{R}enormalization automated by {H}opf algebra}}.
\newblock \btxjournalfont {J. Symb. Comput.}, 27:581, 1999.
\newblock arXiv:hep-th/9810087v1.

\bibitem {tHooft}
\btxnamefont {\btxlastnamefont {{G. 't Hooft}}}\btxauthorcolon\ \btxjtitlefont
  {\btxifchangecase {{R}enormalization of {M}assless {Y}ang-{M}ills
  {F}ields}{{R}enormalization of {M}assless {Y}ang-{M}ills {F}ields}}.
\newblock \btxjournalfont {Nucl. Phys. B}, 33 (1):173--199, 1971.

\bibitem {tHooft_Veltman}
\btxnamefont {\btxlastnamefont {{G. 't Hooft}}} \btxandlong {}\ \btxnamefont
  {\btxlastnamefont {{M. Veltman}}}\btxauthorcolon\ \btxtitlefont
  {{D}iagrammar}, \btxpageslong {}\ 177--322.
\newblock \btxpublisherfont {Springer US}, Boston, MA, 1974\ifbtxprintISBN {,
  \mbox{\btxISBN~\btxISBNfont {978-1-4684-2826-1}}}.

\bibitem {Cvitanovic}
\btxnamefont {\btxlastnamefont {{P. Cvitanovi\'{c}}}}\btxauthorcolon\
  \btxjtitlefont {\btxifchangecase {{F}ield {T}heory}{{F}ield {T}heory}}.
\newblock \btxjournalfont {Nordita Lecture Notes}, 1983.
\newblock Available at \url{http://chaosbook.org/FieldTheory/}.

\bibitem {Faizal}
\btxnamefont {\btxlastnamefont {{M. Faizal}}}\btxauthorcolon\ \btxjtitlefont
  {\btxifchangecase {{BRST} and {A}nti-{BRST} {S}ymmetries in {P}erturbative
  {Q}uantum {G}ravity}{{BRST} and {A}nti-{BRST} {S}ymmetries in {P}erturbative
  {Q}uantum {G}ravity}}.
\newblock \btxjournalfont {Found. Phys.}, 41:270--277, 2011.
\newblock arXiv:1010.1143v2 [gr-qc].

\bibitem {Prinz_5}
\btxnamefont {\btxlastnamefont {{D. Prinz}}}\btxauthorcolon\ \btxjtitlefont
  {\btxifchangecase {{T}he {BRST} {D}ouble {C}omplex for the {C}oupling of
  {G}ravity to {G}auge {T}heories}{{T}he {BRST} {D}ouble {C}omplex for the
  {C}oupling of {G}ravity to {G}auge {T}heories}}.
\newblock \btxjournalfont {Adv. Theor. Math. Phys.}, 29(7):2005--2045, 2025.
\newblock arXiv:2206.00780v2 [hep-th].

\bibitem {Prinz_6}
\btxnamefont {\btxlastnamefont {{D. Prinz}}}\btxauthorcolon\ \btxjtitlefont
  {\btxifchangecase {{S}ymmetric {G}host {L}agrange {D}ensities for the
  {C}oupling of {G}ravity to {G}auge {T}heories}{{S}ymmetric {G}host {L}agrange
  {D}ensities for the {C}oupling of {G}ravity to {G}auge {T}heories}}.
\newblock \btxjournalfont {Accepted for publication in Adv. High Energy Phys.},
  2025.
\newblock arXiv:2207.07593v2 [hep-th].

\bibitem {Prinz_4}
\btxnamefont {\btxlastnamefont {{D. Prinz}}}\btxauthorcolon\ \btxjtitlefont
  {\btxifchangecase {{G}ravity-{M}atter {F}eynman {R}ules for any
  {V}alence}{{G}ravity-{M}atter {F}eynman {R}ules for any {V}alence}}.
\newblock \btxjournalfont {Class. Quantum Grav.}, 38:215003, 2021.
\newblock arXiv:2004.09543v4 [hep-th].

\bibitem {Kissler}
\btxnamefont {\btxlastnamefont {{H. Ki\ss{}ler}}}\btxauthorcolon\
  \btxjtitlefont {\btxifchangecase {{O}ff-shell diagrammatics for quantum
  gravity}{{O}ff-shell diagrammatics for quantum gravity}}.
\newblock \btxjournalfont {Phys. Lett. B}, 816:136219, 2021.
\newblock arXiv:2007.08894v2 [hep-th].

\bibitem {Prinz_7}
\btxnamefont {\btxlastnamefont {{D. Prinz}}}\btxauthorcolon\ \btxtitlefont
  {\btxifchangecase {{T}ransversality in the {C}oupling of {G}ravity to {G}auge
  {T}heories}{{T}ransversality in the {C}oupling of {G}ravity to {G}auge
  {T}heories}}, Preprint: 2022.
\newblock arXiv:2208.14166v1 [hep-th].

\bibitem {Connes_Kreimer_0}
\btxnamefont {\btxlastnamefont {{A. Connes}}} \btxandlong {}\ \btxnamefont
  {\btxlastnamefont {{D. Kreimer}}}\btxauthorcolon\ \btxjtitlefont
  {\btxifchangecase {{R}enormalization in quantum field theory and the
  {R}iemann-{H}ilbert problem}{{R}enormalization in quantum field theory and
  the {R}iemann-{H}ilbert problem}}.
\newblock \btxjournalfont {JHEP}, 09:024, 1999.
\newblock arXiv:hep-th/9909126v3.

\bibitem {Connes_Kreimer_1}
\btxnamefont {\btxlastnamefont {{A. Connes}}} \btxandlong {}\ \btxnamefont
  {\btxlastnamefont {{D. Kreimer}}}\btxauthorcolon\ \btxjtitlefont
  {\btxifchangecase {{R}enormalization in quantum field theory and the
  {R}iemann-{H}ilbert problem {I}: the {H}opf algebra structure of graphs and
  the main theorem}{{R}enormalization in quantum field theory and the
  {R}iemann-{H}ilbert problem {I}: the {H}opf algebra structure of graphs and
  the main theorem}}.
\newblock \btxjournalfont {Commun. Math. Phys.}, 210:249--273, 1999.
\newblock arXiv:hep-th/9912092v1.

\bibitem {Connes_Kreimer_2}
\btxnamefont {\btxlastnamefont {{A. Connes}}} \btxandlong {}\ \btxnamefont
  {\btxlastnamefont {{D. Kreimer}}}\btxauthorcolon\ \btxjtitlefont
  {\btxifchangecase {{R}enormalization in quantum field theory and the
  {R}iemann-{H}ilbert problem {II}: the \(\beta\)-function, diffeomorphisms and
  the renormalization group}{{R}enormalization in quantum field theory and the
  {R}iemann-{H}ilbert problem {II}: the \(\beta\)-function, diffeomorphisms and
  the renormalization group}}.
\newblock \btxjournalfont {Commun. Math. Phys.}, 216:215--241, 2000.
\newblock arXiv:hep-th/0003188v1.

\bibitem {Kreimer_Anatomy}
\btxnamefont {\btxlastnamefont {{D. Kreimer}}}\btxauthorcolon\ \btxjtitlefont
  {\btxifchangecase {{A}natomy of a gauge theory}{{A}natomy of a gauge
  theory}}.
\newblock \btxjournalfont {Annals Phys.}, 321:2757--2781, 2006.
\newblock arXiv:hep-th/0509135v3.

\bibitem {vSuijlekom_QED}
\btxnamefont {\btxlastnamefont {{W. D. van Suijlekom}}}\btxauthorcolon\
  \btxjtitlefont {\btxifchangecase {{T}he {H}opf algebra of {F}eynman graphs in
  {QED}}{{T}he {H}opf algebra of {F}eynman graphs in {QED}}}.
\newblock \btxjournalfont {Lett. Math. Phys.}, 77:265--281, 2006.
\newblock arXiv:hep-th/0602126v2.

\bibitem {vSuijlekom_QCD}
\btxnamefont {\btxlastnamefont {{W. D. van Suijlekom}}}\btxauthorcolon\
  \btxjtitlefont {\btxifchangecase {{R}enormalization of gauge fields: {A}
  {H}opf algebra approach}{{R}enormalization of gauge fields: {A} {H}opf
  algebra approach}}.
\newblock \btxjournalfont {Commun. Math. Phys.}, 276:773--798, 2007.
\newblock arXiv:hep-th/0610137v1.

\bibitem {vSuijlekom_BV}
\btxnamefont {\btxlastnamefont {{W. D. van Suijlekom}}}\btxauthorcolon\
  \btxjtitlefont {\btxifchangecase {{T}he structure of renormalization {H}opf
  algebras for gauge theories {I}: {R}epresenting {F}eynman graphs on
  {BV}-algebras}{{T}he structure of renormalization {H}opf algebras for gauge
  theories {I}: {R}epresenting {F}eynman graphs on {BV}-algebras}}.
\newblock \btxjournalfont {Commun. Math. Phys.}, 290:291--319, 2009.
\newblock arXiv:0807.0999v2 [math-ph].

\bibitem {Prinz_3}
\btxnamefont {\btxlastnamefont {{D. Prinz}}}\btxauthorcolon\ \btxjtitlefont
  {\btxifchangecase {{G}auge {S}ymmetries and {R}enormalization}{{G}auge
  {S}ymmetries and {R}enormalization}}.
\newblock \btxjournalfont {Math. Phys. Anal. Geom.}, 25(3):20, 2022.
\newblock arXiv:2001.00104v4 [math-ph].

\bibitem {Kreimer_QG1}
\btxnamefont {\btxlastnamefont {{D. Kreimer}}}\btxauthorcolon\ \btxjtitlefont
  {\btxifchangecase {{A} remark on quantum gravity}{{A} remark on quantum
  gravity}}.
\newblock \btxjournalfont {Annals Phys.}, 323:49--60, 2008.
\newblock arXiv:0705.3897v1 [hep-th].

\bibitem {Prinz_PhD}
\btxnamefont {\btxlastnamefont {{D. Prinz}}}\btxauthorcolon\ \btxtitlefont
  {{R}enormalization of {G}auge {T}heories and {G}ravity}.
\newblock \btxphdthesis {}, Humboldt University of Berlin, 2022.
\newblock Available at \url{https://doi.org/10.18452/25401} and
  arXiv:2210.17510v1 [hep-th].

\bibitem {Kreimer_Yeats}
\btxnamefont {\btxlastnamefont {{D. Kreimer}}} \btxandlong {}\ \btxnamefont
  {\btxlastnamefont {{K. Yeats}}}\btxauthorcolon\ \btxjtitlefont
  {\btxifchangecase {{P}roperties of the corolla polynomial of a 3-regular
  graph}{{P}roperties of the corolla polynomial of a 3-regular graph}}.
\newblock \btxjournalfont {The electronic journal of combinatorics}, 20(1),
  2012.
\newblock arXiv:1207.5460v1 [math.CO].

\bibitem {Kreimer_Sars_vSuijlekom}
\btxnamefont {\btxlastnamefont {{D. Kreimer}}}, \btxnamefont {\btxlastnamefont
  {{M. Sars}}}\btxandcomma {} \btxandlong {}\ \btxnamefont {\btxlastnamefont
  {{W. D. van Suijlekom}}}\btxauthorcolon\ \btxjtitlefont {\btxifchangecase
  {{Q}uantization of gauge fields, graph polynomials and graph
  cohomology}{{Q}uantization of gauge fields, graph polynomials and graph
  cohomology}}.
\newblock \btxjournalfont {Annals Phys.}, 336:180--222, 2013.
\newblock arXiv:1208.6477v4 [hep-th].

\bibitem {Sars_PhD}
\btxnamefont {\btxlastnamefont {{M. Sars}}}\btxauthorcolon\ \btxtitlefont
  {{P}arametric {R}epresentation of {F}eynman {A}mplitudes in {G}auge
  {T}heories}.
\newblock \btxphdthesis {}, Humboldt University of Berlin, 2015.
\newblock Available at \url{https://dx.doi.org/10.18452/17302}.

\bibitem {Prinz_1}
\btxnamefont {\btxlastnamefont {{D. Prinz}}}\btxauthorcolon\ \btxjtitlefont
  {\btxifchangecase {{T}he {C}orolla {P}olynomial for {S}pontaneously {B}roken
  {G}auge {T}heories}{{T}he {C}orolla {P}olynomial for {S}pontaneously {B}roken
  {G}auge {T}heories}}.
\newblock \btxjournalfont {Math. Phys. Anal. Geom.}, 19(3):18, 2016.
\newblock arXiv:1603.03321v3 [math-ph].

\bibitem {Golz_PhD}
\btxnamefont {\btxlastnamefont {{M. Golz}}}\btxauthorcolon\ \btxtitlefont
  {{P}arametric quantum electrodynamics}.
\newblock \btxphdthesis {}, Humboldt University of Berlin, 2019.
\newblock Available at \url{https://dx.doi.org/10.18452/19776}.

\bibitem {Berghoff_Knispel}
\btxnamefont {\btxlastnamefont {{M. Berghoff, A. Knispel}}}\btxauthorcolon\
  \btxjtitlefont {\btxifchangecase {{C}omplexes of marked graphs in gauge
  theory}{{C}omplexes of marked graphs in gauge theory}}.
\newblock \btxjournalfont {Lett. Math. Phys.}, 110:2417--2433, 2020.
\newblock arXiv:1908.06640v2 [math-ph].

\bibitem {Becchi_Rouet_Stora_1}
\btxnamefont {\btxlastnamefont {{C. Becchi, A. Rouet and R.
  Stora}}}\btxauthorcolon\ \btxjtitlefont {\btxifchangecase {{T}he abelian
  {H}iggs {K}ibble model, unitarity of the {S}-operator}{{T}he abelian {H}iggs
  {K}ibble model, unitarity of the {S}-operator}}.
\newblock \btxjournalfont {Phys. Lett. B}, 52:344--346, 1974.

\bibitem {Becchi_Rouet_Stora_2}
\btxnamefont {\btxlastnamefont {{C. Becchi, A. Rouet and R.
  Stora}}}\btxauthorcolon\ \btxjtitlefont {\btxifchangecase {{R}enormalization
  of the abelian {H}iggs-{K}ibble model}{{R}enormalization of the abelian
  {H}iggs-{K}ibble model}}.
\newblock \btxjournalfont {Commun. Math. Phys.}, 42:127--162, 1975.

\bibitem {Becchi_Rouet_Stora_3}
\btxnamefont {\btxlastnamefont {{C. Becchi, A. Rouet and R.
  Stora}}}\btxauthorcolon\ \btxjtitlefont {\btxifchangecase {{R}enormalization
  of gauge theories}{{R}enormalization of gauge theories}}.
\newblock \btxjournalfont {Annals Phys.}, 98:287--321, 1976.

\bibitem {Tyutin}
\btxnamefont {\btxlastnamefont {{I. V. Tyutin}}}\btxauthorcolon\ \btxjtitlefont
  {\btxifchangecase {{G}auge {I}nvariance in {F}ield {T}heory and {S}tatistical
  {P}hysics in {O}perator {F}ormalism}{{G}auge {I}nvariance in {F}ield {T}heory
  and {S}tatistical {P}hysics in {O}perator {F}ormalism}}.
\newblock \btxjournalfont {Lebedev Physics Institute preprint}, 39, 1975.
\newblock arXiv:0812.0580v2 [hep-th].

\bibitem {Batalin_Vilkovisky_1}
\btxnamefont {\btxlastnamefont {{I. A. Batalin}}} \btxandlong {}\ \btxnamefont
  {\btxlastnamefont {{G. A. Vilkovisky}}}\btxauthorcolon\ \btxjtitlefont
  {\btxifchangecase {{G}auge {A}lgebra and {Q}uantization}{{G}auge {A}lgebra
  and {Q}uantization}}.
\newblock \btxjournalfont {Phys. Lett. B. 102 (1): 27–31}, 1981.

\bibitem {Batalin_Vilkovisky_2}
\btxnamefont {\btxlastnamefont {{I. A. Batalin}}} \btxandlong {}\ \btxnamefont
  {\btxlastnamefont {{G. A. Vilkovisky}}}\btxauthorcolon\ \btxjtitlefont
  {\btxifchangecase {{Q}uantization of {G}auge {T}heories with {L}inearly
  {D}ependent {G}enerators}{{Q}uantization of {G}auge {T}heories with
  {L}inearly {D}ependent {G}enerators}}.
\newblock \btxjournalfont {Physical Review D. 28 (10): 2567–2582}, 1983.

\bibitem {tHooft_Veltman_Nobel_Prize}
\btxnamefont {\btxlastnamefont {{G. 't Hooft}}} \btxandlong {}\ \btxnamefont
  {\btxlastnamefont {{M. J. G.Veltman}}}\btxauthorcolon\ \btxjtitlefont
  {\btxifchangecase {{R}egularization and {R}enormalization of {G}auge
  {F}ields}{{R}egularization and {R}enormalization of {G}auge {F}ields}}.
\newblock \btxjournalfont {Nucl. Phys. B}, 44:189--213, 1972.

\bibitem {Taylor}
\btxnamefont {\btxlastnamefont {{J. C. Taylor}}}\btxauthorcolon\ \btxjtitlefont
  {\btxifchangecase {{W}ard identities and charge renormalization of the
  {Y}ang-{M}ills field}{{W}ard identities and charge renormalization of the
  {Y}ang-{M}ills field}}.
\newblock \btxjournalfont {Nucl. Phys. B}, 33 (2):436--444, 1971.

\bibitem {Slavnov}
\btxnamefont {\btxlastnamefont {{A. A. Slavnov}}}\btxauthorcolon\
  \btxjtitlefont {\btxifchangecase {{W}ard identities in gauge theories}{{W}ard
  identities in gauge theories}}.
\newblock \btxjournalfont {Theor. Math. Phys.}, 10 (2):99--104, 1972.

\bibitem {Ward}
\btxnamefont {\btxlastnamefont {{J. C. Ward}}}\btxauthorcolon\ \btxjtitlefont
  {\btxifchangecase {{A}n {I}dentity in {Q}uantum {E}lectrodynamics}{{A}n
  {I}dentity in {Q}uantum {E}lectrodynamics}}.
\newblock \btxjournalfont {Phys. Rev.}, 78:182, 1950.

\bibitem {Takahashi}
\btxnamefont {\btxlastnamefont {{Y. Takahashi}}}\btxauthorcolon\ \btxjtitlefont
  {\btxifchangecase {{O}n the {G}eneralized {W}ard {I}dentity}{{O}n the
  {G}eneralized {W}ard {I}dentity}}.
\newblock \btxjournalfont {Nuovo Cim.}, 6:371, 1957.

\bibitem {Prinz_2}
\btxnamefont {\btxlastnamefont {{D. Prinz}}}\btxauthorcolon\ \btxjtitlefont
  {\btxifchangecase {{A}lgebraic {S}tructures in the {C}oupling of {G}ravity to
  {G}auge {T}heories}{{A}lgebraic {S}tructures in the {C}oupling of {G}ravity
  to {G}auge {T}heories}}.
\newblock \btxjournalfont {Annals Phys.}, 426:168395, 2021.
\newblock arXiv:1812.09919v3 [hep-th].

\bibitem {Prinz_8}
\btxnamefont {\btxlastnamefont {{D. Prinz}}}\btxauthorcolon\ \btxtitlefont
  {\btxifchangecase {{O}n {P}erturbative {Q}uantum {G}ravity with a
  {C}osmological {C}onstant}{{O}n {P}erturbative {Q}uantum {G}ravity with a
  {C}osmological {C}onstant}}, Preprint: 2023.
\newblock arXiv:2303.14160v1 [hep-th].

\bibitem {Guo}
\btxnamefont {\btxlastnamefont {{L. Guo}}}\btxauthorcolon\ \btxjtitlefont
  {\btxifchangecase {{A}lgebraic {B}irkhoff decomposition and its
  applications}{{A}lgebraic {B}irkhoff decomposition and its applications}}.
\newblock \btxjournalfont {Automorphic Forms and the Langlands Program},
  International Press:283--323, 2008.
\newblock arXiv:0807.2266v1 [math.RA].

\bibitem {Panzer}
\btxnamefont {\btxlastnamefont {{E. Panzer}}}\btxauthorcolon\ \btxtitlefont
  {\btxifchangecase {{H}opf algebraic {R}enormalization of {K}reimer's toy
  model}{{H}opf algebraic {R}enormalization of {K}reimer's toy model}}, 2012.
\newblock arXiv:1202.3552v1 [math.QA].

\bibitem {Yeats_PhD}
\btxnamefont {\btxlastnamefont {{K. Yeats}}}\btxauthorcolon\ \btxtitlefont
  {{G}rowth estimates for {D}yson-{S}chwinger equations}.
\newblock \btxphdthesis {}, Boston University, Graduate School of Arts and
  Sciences, 2008.
\newblock arXiv:0810.2249v1 [math-ph].

\bibitem {Borinsky_Feyngen}
\btxnamefont {\btxlastnamefont {{M. Borinsky}}}\btxauthorcolon\ \btxjtitlefont
  {\btxifchangecase {{F}eynman graph generation and calculations in the {H}opf
  algebra of {F}eynman graphs}{{F}eynman graph generation and calculations in
  the {H}opf algebra of {F}eynman graphs}}.
\newblock \btxjournalfont {Comput. Phys. Commun.}, 185:3317--3330, 2014.
\newblock arXiv:1402.2613v2 [hep-th].

\bibitem {vSuijlekom_Multiplicative}
\btxnamefont {\btxlastnamefont {{W. D. van Suijlekom}}}\btxauthorcolon\
  \btxjtitlefont {\btxifchangecase {{M}ultiplicative renormalization and {H}opf
  algebras}{{M}ultiplicative renormalization and {H}opf algebras}}.
\newblock \btxjournalfont {Arithmetic and geometry around quantization. Eds. O.
  Ceyhan, Yu.-I. Manin and M. Marcolli}, Progress in Mathematics 279, 2010.
\newblock arXiv:0707.0555v1 [hep-th].

\bibitem {Kissler_Kreimer}
\btxnamefont {\btxlastnamefont {{H. Ki\ss{}ler}}} \btxandlong {}\ \btxnamefont
  {\btxlastnamefont {{D. Kreimer}}}\btxauthorcolon\ \btxjtitlefont
  {\btxifchangecase {{D}iagrammatic {C}ancellations and the {G}auge
  {D}ependence of {QED}}{{D}iagrammatic {C}ancellations and the {G}auge
  {D}ependence of {QED}}}.
\newblock \btxjournalfont {Phys. Lett. B}, 764:318--321, 2017.
\newblock arXiv:1607.05729v4 [hep-th].

\bibitem {Gracey_Kissler_Kreimer}
\btxnamefont {\btxlastnamefont {{J. A. Gracey}},~{H. Kißler}, {D.
  Kreimer}}\btxauthorcolon\ \btxjtitlefont {\btxifchangecase {{O}n the
  self-consistency of off-shell {S}lavnov-{T}aylor identities in {QCD}}{{O}n
  the self-consistency of off-shell {S}lavnov-{T}aylor identities in {QCD}}}.
\newblock \btxjournalfont {Phys. Rev. D}, 100(8):085001, 2019.
\newblock arXiv:1906.07996v2 [hep-th].

\bibitem {Harlander_Klein_Lipp}
\btxnamefont {\btxlastnamefont {{R. V. Harlander}}}, \btxnamefont
  {\btxlastnamefont {{S. Y. Klein}}}\btxandcomma {} \btxandlong {}\
  \btxnamefont {\btxlastnamefont {{M. Lipp}}}\btxauthorcolon\ \btxjtitlefont
  {\btxifchangecase {{FeynGame}}{{FeynGame}}}.
\newblock \btxjournalfont {Comput. Phys. Commun.}, 256:107465, 2020.
\newblock arXiv:2003.00896v1 [physics.ed-ph].

\bibitem {Harlander_Klein_Schaaf}
\btxnamefont {\btxlastnamefont {{R. V. Harlander}}}, \btxnamefont
  {\btxlastnamefont {{S. Y. Klein}}}\btxandcomma {} \btxandlong {}\
  \btxnamefont {\btxlastnamefont {{M. C. Schaaf}}}\btxauthorcolon\
  \btxjtitlefont {\btxifchangecase {{FeynGame}-2.1 -- {F}eynman diagrams made
  easy}{{FeynGame}-2.1 -- {F}eynman diagrams made easy}}.
\newblock \btxjournalfont {PoS}, EPS-HEP2023:657, 2024.
\newblock arXiv:2401.12778v1 [hep-ph].

\bibitem {Bundgen_Harlander_Klein_Schaaf}
\btxnamefont {\btxlastnamefont {{L. B{\"u}ndgen}}}, \btxnamefont
  {\btxlastnamefont {{R. V. Harlander}}}, \btxnamefont {\btxlastnamefont {{S.
  Y. Klein}}}\btxandcomma {} \btxandlong {}\ \btxnamefont {\btxlastnamefont
  {{M. C. Schaaf}}}\btxauthorcolon\ \btxjtitlefont {\btxifchangecase
  {{FeynGame} 3.0}{{FeynGame} 3.0}}.
\newblock \btxjournalfont {Comput. Phys. Commun.}, 314:109662, 2025.
\newblock arXiv:2501.04651v1 [hep-ph].

\bibitem {Goldberg}
\btxnamefont {\btxlastnamefont {{J. N. Goldberg}}}\btxauthorcolon\
  \btxjtitlefont {\btxifchangecase {{C}onservation {L}aws in {G}eneral
  {R}elativity}{{C}onservation {L}aws in {G}eneral {R}elativity}}.
\newblock \btxjournalfont {Phys. Rev.}, 111:315--320, 1958.

\bibitem {Capper_Leibbrandt_Ramon-Medrano}
\btxnamefont {\btxlastnamefont {{D. M. Capper}}}, \btxnamefont
  {\btxlastnamefont {{G. Leibbrandt}}}\btxandcomma {} \btxandlong {}\
  \btxnamefont {\btxlastnamefont {{M. Ram\'{o}n Medrano}}}\btxauthorcolon\
  \btxjtitlefont {\btxifchangecase {{C}alculation of the graviton selfenergy
  using dimensional regularization}{{C}alculation of the graviton selfenergy
  using dimensional regularization}}.
\newblock \btxjournalfont {Phys. Rev. D}, 8:4320--4331, 1973.

\bibitem {Capper_Medrano}
\btxnamefont {\btxlastnamefont {{D. M. Capper}}} \btxandlong {}\ \btxnamefont
  {\btxlastnamefont {{M. Ram\'{o}n Medrano}}}\btxauthorcolon\ \btxjtitlefont
  {\btxifchangecase {{G}ravitational {S}lavnov--{W}ard
  identities}{{G}ravitational {S}lavnov--{W}ard identities}}.
\newblock \btxjournalfont {Phys. Rev. D}, 9:1641--1647, 1974.

\bibitem {Capper_Namazie}
\btxnamefont {\btxlastnamefont {{D. M. Capper}}} \btxandlong {}\ \btxnamefont
  {\btxlastnamefont {{M. A. Namazie}}}\btxauthorcolon\ \btxjtitlefont
  {\btxifchangecase {{A} {G}eneral {G}auge {C}alculation of the {G}raviton
  {S}elfenergy}{{A} {G}eneral {G}auge {C}alculation of the {G}raviton
  {S}elfenergy}}.
\newblock \btxjournalfont {Nucl. Phys. B}, 142:535--547, 1978.

\end{thebibliography}
\bibliographystyle{babunsrt}

\end{document}